\documentclass{article}
\usepackage{fullpage}

\usepackage[usenames,dvipsnames]{color}
\usepackage{xfrac}
\usepackage{xspace}

\newcommand{\half}{{\sfrac 1 2}\xspace}

\usepackage{graphicx,amsmath,amssymb,amsthm}
\usepackage{algorithmicx,algorithm}
\usepackage[noend]{algpseudocode}

\algnewcommand\algorithmicinput{\textbf{INPUT:}}
\algnewcommand\INPUT{\item[\algorithmicinput]}
\algnewcommand\algorithmicoutput{\textbf{OUTPUT:}}
\algnewcommand\OUTPUT{\item[\algorithmicoutput]}

 \newtheorem{definition}{Definition}
 
 \newtheorem{lemma}{Lemma}
 \newtheorem{theorem}{Theorem}
\newtheorem{claim}{Claim}
 \newtheorem{corollary}{Corollary}

\newcommand{\R}{\mathbb{R}}

\newcommand{\cF}{\ensuremath{\mathcal F}\xspace}

\newcommand{\Z}{\ensuremath{\mathbb Z}\xspace}

\newcommand{\Oh}{O}

\newcommand{\parameterizedproblem}[4]{
\smallskip
\noindent\fbox{\begin{minipage}{\linewidth}
\noindent  \textsc{#1} \hfill \textbf{Parameter:} #3\\
\textbf{Input:} #2\\ 
\textbf{Question:} #4
\end{minipage}}
}

\newcommand{\maybeqed}{}

\begin{document}
\title{Half-integrality, LP-branching and FPT Algorithms\thanks{A preliminary version of this paper appeared in
the proceedings of SODA 2014.}}

\author{Yoichi Iwata\thanks{Department of Computer Science,
    Graduate School of Information Science and Technology,
    The University of Tokyo.
    \texttt{y.iwata@is.s.u-tokyo.ac.jp}
    Supported by Grant-in-Aid for JSPS Fellows (256487).
  }
  \and
  Magnus Wahlstr\"om\thanks{Royal Holloway, University of London. \texttt{magnus.wahlstrom@rhul.ac.uk}}
  \and
  Yuichi Yoshida\thanks{
    National Institute of Informatics, and Preferred Infrastructure, Inc.
    \texttt{yyoshida@nii.ac.jp}
    Supported by JSPS Grant-in-Aid for Research Activity Start-up (No. 24800082), MEXT Grant-in-Aid for Scientific Research on Innovative Areas (No. 24106003), and JST, ERATO, Kawarabayashi Large Graph Project.
  }
}

\maketitle

\begin{abstract}
A recent trend in parameterized algorithms is the application of polytope
tools (specifically, LP-branching) to FPT algorithms (e.g., Cygan~et~al.,
2011; Narayanaswamy~et~al., 2012). Though the list of work in this
direction is short, the results are already interesting, yielding
significant speedups for a range of important problems. However, the
existing approaches require the underlying polytope to have very
restrictive properties, including half-integrality and 
Nemhauser-Trotter-style persistence properties.  
To date, these properties are essentially known to hold only for two
classes of polytopes, covering the cases of \textsc{Vertex Cover}
(Nemhauser and Trotter, 1975) and \textsc{Node Multiway Cut} (Garg~et~al.,
1994). 

Taking a slightly different approach, we view half-integrality as a
\emph{discrete} relaxation of a problem, e.g., a relaxation of the search
space from $\{0,1\}^V$ to $\{0,\half,1\}^V$ such that the new problem admits
a polynomial-time exact solution. Using tools from CSP (in particular
Thapper and \v{Z}ivn\'y, 2012) to study the existence of such relaxations,  
we are able to provide a much broader class of half-integral polytopes
with the required properties.

Our results unify and significantly extend the previously known cases. In
addition to the new insight into problems with half-integral relaxations,
our results yield a range of new and improved FPT algorithms, including
an $O^*(|\Sigma|^{2k})$-time algorithm for node-deletion \textsc{Unique
  Label Cover} with label set $\Sigma$ (improving the previous bound of
$O^*(|\Sigma|^{O(k^2 \log k)})$ due to Chitnis~et~al., 2012)
and an $O^*(4^k)$-time algorithm for \textsc{Group Feedback Vertex Set},
including the setting where the group is only given by oracle access
(improving on the previous bound of $O^*(2^{O(k\log k)})$ due to Cygan et
al., 2012). 
The latter bound is 
optimal under the Exponential Time Hypothesis.
The latter result also implies the first single-exponential time FPT
algorithm for \textsc{Subset Feedback Vertex Set}, answering an open
question of Cygan~et~al. (2012). 
Additionally, we propose a network flow-based approach to solve some cases of the relaxation problem.
This gives the first linear-time FPT algorithm to edge-deletion \textsc{Unique Label Cover}.

Interestingly, despite the half-integrality, our result do not imply any
approximation results (as may be expected, given the \textsc{Unique
  Games}-hardness of the covered problems).  
\end{abstract}

%\begin{keywords} 
%  Fixed parameter tractability, $k$-submodularity.
%\end{keywords}
%
%\begin{AMS}
%   05C85
%\end{AMS}

%\newpage

\section{Introduction}

Polytope methods, and methods related to linear and integer programming in
general, have been hugely successful in combinatorial optimisation, both
for deriving exact polynomial-time results and for purposes of
approximation (see, e.g., the book of Schrijver~\cite{SchrijverBook}). 
However, the methods have seen less application for questions of getting
faster exact (i.e., non-approximate) solutions to NP-hard problems,
at least from a theoretical perspective. 
(Industrial mixed integer programming-solvers such as CPLEX, though
frequently efficient, are not our concern here since usually, no
non-trivial performance guarantees are known.)

A few such applications have emerged in recent years in the field of
parameterized complexity; specifically, two sets of problems --
\textsc{Node Multiway Cut}~\cite{CyganPPW13MWC} and problems related to
\textsc{Vertex Cover}~\cite{NarayanaswamyRRS12,LokshtanovNRRS12CoRR} -- have
been shown to be FPT parameterized by the \emph{above LP} parameter, i.e.,
given an instance of one of these problems, it can be decided in
$O^*(4^k)$ time whether there is a solution that is at most $k$ points
more expensive than the LP-optimum. In the former case, due to the
integrality gap of the \textsc{Multiway Cut} LP~\cite{GargVY04}, this
results in an $O^*(2^k)$-time FPT algorithm for the natural
parameterization of the problem, improving on previous results of
$O^*(4^k)$; in the latter case, through parameter-preserving problem
reductions, the result is improved FPT algorithms for a range of problems
(e.g., problems expressible in \textsc{Almost 2-SAT}, a.k.a., 2-CNF deletion). 

However, despite the promise of the approach (and the programmatic view
taken in the latter set of papers~\cite{NarayanaswamyRRS12,LokshtanovNRRS12CoRR}), 
we still know only few such applications. (Also note that if $k$ is taken
as the above ``gap'' parameter, then in general it would be NP-hard to
decide whether $k=0$.) Furthermore, an inspection of the tools used reveal
that the methods are quite similar, and very specific; it is a matter of
FPT applications of the half-integrality results of Nemhauser and
Trotter~\cite{NemhauserT75} in the latter case, and similar
half-integrality results for \textsc{Node Multiway Cut} in the former
case, as shown by Garg~et~al.~\cite{GargVY04} and refined for FPT purposes
by Guillemot~\cite{Guillemot11} and Cygan~et~al.~\cite{CyganPPW13MWC}. 
Therefore, a good first step towards a better understanding of the power
of LP-relaxations for FPT problems (or vice versa, e.g., to further the
parameterized study of mixed integer programming) seems to be to consider
specifically the property of half-integrality. 

\subsection{Integral and half-integral polytopes}
Compared to our knowledge about integral polytopes (e.g., connections to
totally unimodular matrices and the notion of total dual integrality), 
our knowledge of half-integrality seems rather more spotty. 
It seems that most of what is available can be enumerated as a few quick
examples, e.g., the above-mentioned cases of \textsc{Vertex
  Cover}~\cite{NemhauserT75} and \textsc{Node Multiway
  Cut}~\cite{GargVY04}; Hochbaum's IP2 programs~\cite{Hochbaum02}; 
and a few related cases, such as the continuous relaxation for
\textsc{Submodular Vertex Cover}~\cite{IwataN09}.
Of these, probably the most ambitious study of half-integrality is the
work of Hochbaum~\cite{Hochbaum02}, where a general IP of a certain
restricted form is shown to admit half-integral solutions. 
Still, of the applications mentioned in~\cite{Hochbaum02}, most if not all
(e.g., all applications with a Boolean domain)
can be covered by a simple reduction to \textsc{Almost 2-SAT}. 
One should also mention Kolmogorov~\cite{Kolmogorov12DAM}; see below.

One important note is that half-integrality is more specific than having
an integrality gap of $2$. While the latter clearly implies the same
approximation result, half-integrality imposes much more structure on the
solutions of a problem (as seen, e.g., by the FPT applications above and
in the rest of this paper). Examples of LP-relaxations which are
2-approximate but not half-integral would include \textsc{Multicut in
  Trees}~\cite{GargVY97} and \textsc{Feedback Vertex
  Set}~\cite{ChudakGHW98}; see also results achieved via iterative
rounding~\cite{Jain01,FleischerJW06}, e.g., for \textsc{Steiner Tree}. 
In the present paper, we ignore such results, and focus on the topic of
half-integrality. 
%While it is certainly an interesting question for the future whether such
%more general methods can be brought to bear for FPT purposes, in the
%present paper we focus specifically on half-integrality. 
%(Note, however, that we have some results for
%\textsc{Feedback Vertex Set}, via a quite different formulation; see
%Section~\ref{section:gfvs-implies}.) 

In this work, to discover half-integral relaxations, we take a slightly
different approach to the problem from most of the above, inspired rather
by the work of Kolmogorov~\cite{Kolmogorov12DAM}. In essence, we
start from the observation that a half-integral relaxation, unlike a
generic 2-approximate LP-relaxation, actually defines a polynomial-time
solvable problem on a \emph{discrete} search space of $\{0,\half,1\}^n$. 
Thus, we argue that the search for half-integral relaxations, and even for
half-integral polytopes, would benefit from the application of tools
designed to characterise exactly solvable problems, e.g., tools from the
study of constraint satisfaction problems. 

\subsection{CSPs and LP-relaxations}
Constraint satisfaction problems (CSPs) make for a general setting in
which the complexity of various problems can be studied in a systematic
way. In the most common setting, one studies generalisations of SAT: Given
a (one-time fixed) set $\Gamma$ of relation types, what is the complexity
of deciding the satisfiability of a formula which consists of a
conjunction of applications of relations $R \in \Gamma$? For example, by
fixing the domain to be Boolean, and letting $\Gamma$ contain all
3-clauses, one would encode the problem 3-SAT.

For optimisation problems, a generalisation of \emph{valued} CSPs (VCSPs)
has been proposed. Roughly, in this setting, instead of using relations,
one fixes a set $\cF$ of \emph{cost functions}; an instance consists of a
set of applications of functions $f_i \in \cF$, and the task is to
minimise (or maximise) the sum of the values of the functions in the
input. 
One particular case (which has been studied extensively in approximation)
is when the cost functions all take values $0$ and $1$ only, thus encoding
a ``soft version'' of a constraint; e.g., $f(u,v)=[u=v=0]$ (taking cost $1$ if
$u=v=0$, cost $0$ otherwise) would be the soft version of a constraint
$(u \lor v)$.  (In some approximation literature the maximisation
version of VCSP for such soft versions of constraints is taken as the
definition of the CSP problem itself.) 
Again, the interest is in identifying which sets $\cF$ of cost functions
imply polynomial-time solvable versus NP-hard problems, or more
closely what approximation properties the resulting CSP would have. 

The use of various relaxations has been of critical importance to the
solutions for these problems. For approximation, the best results have
been attained using SDP relaxations, and
Raghavendra~\cite{Raghavendra08} showed that assuming the unique games
conjecture~\cite{Khot02}, a particular SDP relaxation achieves the
optimal approximation ratio for every Max CSP problem. However, for the
question of whether finding an \emph{exact} solution is in P or NP-hard, 
it turns out, somewhat surprisingly, that it suffices to use a simple LP-relaxation (known as
the \emph{basic LP}, being essentially a simpler version of the
appropriate level in the Sherali-Adams hierarchy).

To be precise, it follows from a sequence of work by Thapper and
\v{Z}ivn\'y and by Kolmogorov~\cite{ThapperZ12FOCS,Kolmogorov12CoRR,ThapperZ13STOC} 
that for every set of finite-valued cost functions $\cF$, either the basic LP solves
the resulting VCSP exactly, or the VCSP problem is APX-hard. 
Thus, despite our excursion into CSPs, the connection to LP-relaxations
and polytope theory remains, in particular as the LP-relaxation remains
the only known method of solving the problem for several of the covered
problem classes. 

%polytope connection remains, in that for several of the problem classes
%covered, the LP-relaxation remains the only known way to solve the problem
%in polynomial time (e.g., no combinatorial algorithms are known, nor any
%algorithms that work in a value oracle model). 

%not even strongly polynomial? maybe?

%although the power of LP- and SDP-relaxations differ significantly with
%respect to approximation (e.g., maxcut)
%it turns out that for providing \emph{exact} solutions, SDP add no further
%power 

Our application of this framework takes the following shape. 
Assume an NP-hard VCSP problem, defined by a class of cost functions~$\cF$
on a finite domain $D$ (i.e., the search space of the problem is $D^n$).
If our problem has a half-integral LP-relaxation, then there should also
exist a class $\cF'$ of ``relaxed'' versions of the cost functions,
working in a search space $(D')^n$ (e.g., $D'$ would be $D$ extended by
the half-integral values), such that $\cF'$ defines a polynomial-time
solvable problem. We call such a class $\cF'$ a \emph{discrete relaxation}
of the original problem, and refer to values from the original domain $D$
(e.g., $\{0,1\}$) as \emph{integral} values, and values from $D' \setminus
D$ (e.g., \half) as \emph{relaxed} values.
(We also need some technical requirements; 
%To ensure that the relaxation is computationally useful, we need some
%further conditions, including a persistence-like property akin to the
%Nemhauser-Trotter theorem for \textsc{Vertex Cover}; 
see Section~\ref{section:discrete-relaxation}.) % for details.

Assuming that such a discrete relaxation $\cF'$ is found, we may then use
an algorithm, akin to the LP-branching algorithms
of~\cite{CyganPPW13MWC,NarayanaswamyRRS12,LokshtanovNRRS12CoRR}, to solve our
original problem in FPT time, parameterized
%. (We furthermore get, as above, that the
%problem is FPT when parameterized 
by the size of the \emph{relaxation gap}.
The connection to half-integrality lies in the basic LP of the relaxed
class $\cF'$; in our examples, $\cF'$ is a half-integral relaxation of
$\cF$, and the basic LP can be used to construct a simpler LP-relaxation
for the original problem, which then is found to be half-integral.

%
%``half'' from ``relaxation factor'' from $f$ to $f'$
%
%for example bisubmodular functions 

\subsection{Our results}
We show that many known half-integrality results, and several new ones,
can be explained by applying the above framework using the class of 
\emph{$k$-submodular} functions as discrete relaxations.
This includes the above cases of \textsc{(Submodular Cost) Vertex Cover},
\textsc{Almost 2-SAT}, and \textsc{Node Multiway Cut}, as well as a
further generalisation of the first two called \textsc{Bisubmodular Cost
  2-SAT}. In addition, we construct new, possibly unexpected half-integral
LP-relaxations for the \textsc{Group Feedback Vertex Set} and
\textsc{Unique Label Cover} problems, leading to significantly improved
FPT algorithms; see below.

%Furthermore, we show how this perspective provides new half-integral
%relaxations of \textsc{Group Feedback Vertex Set} and (both the node- and
%edge-deletion versions of) \textsc{Unique Label Cover} (see below for
%problem descriptions), leading to significantly improved FPT-algorithms
%for both problems. 

%The main concrete results of this paper stem from the investigation of
%the class of $k$-submodular functions from a perspective of discrete
%relaxation. These are functions defined on a domain of $\{0,1,\ldots,k\}^n$, 
%which generalise submodular and bisubmodular functions (which correspond
%to the cases of $k=1$ and $k=2$, respectively); they are known to admit
%polynomial-time optimisation via the basic LP \todo{TERM USED? or: via an
%  integral LP-relaxation?} due to Thapper and \v{Z}ivn\'y~\cite{ThapperZ12FOCS},
%but no \discuss{``other'' method is known (oracle, combinatorial, strongly
%  pol).}
%
%These functions were \todo{introduced in ... for ...?}
%
%We show that if the value $0$ is treated as a ``relaxed'' value, and the
%values $\{1,\ldots,k\}$ as the ...
%
%persistent discrete relaxation
%
%we characterise the class of ``simple soft constraints'' ? no? we
%characterise ... binary? the solution structure itself?
%
%broad class of relevant problems
%

The framework immediately implies an integral LP-formulation of the
half-integral relaxations of the above-mentioned problems (i.e., an
integral polytope over a larger set of variables); however, the
resulting formulation has for many problems an inconveniently large
dimension, preventing it to be used in full generality. To work around
this problem, we construct an alternative, half-integral LP-relaxation
with fewer variables, inspired by the basic LP and the construction in~\cite{GargVY04}. 

\textsc{Unique Label Cover} is the problem which lies at the heart of the
unique games conjecture~\cite{Khot02}, which is of central importance to
the field of approximation algorithms. Previous work by Chitnis et
al.~\cite{ChitnisCHPP12} gave an $O^*(|\Sigma|^{O(p^2 \log p)})$-time FPT
algorithm for the problem using a highly involved probabilistic approach
(here, $\Sigma$ is the label set, and $p$ is the solution cost). 
Via our new LP-relaxation, we solve the problem in time
$O^*(|\Sigma|^{2p})$, for both the edge- and vertex-deletion versions;
furthermore, our result is deterministic.

\textsc{Group Feedback Vertex Set} (GFVS) is a powerful generalisation of
\textsc{Feedback Vertex Set} and \textsc{Odd Cycle Transversal}; we refer
to Section~\ref{section:gfvs} and the cited literature for details. 
The FPT study of this problem was initiated by
Guillemot~\cite{Guillemot11}; Cygan~et~al.~\cite{CyganPP12} showed that
the problem is FPT in a very general form (technically, when the input
provides only black-box oracle access to the group), 
with a running time of $O^*(2^{O(k\log k)})$. They note that in this
general form, GFVS subsumes \textsc{Subset Feedback Vertex Set}, for which
an~$O^*(2^{O(k \log k)})$-time algorithm was previously given~\cite{CyganPPW13}. 
They note that their running time seems difficult to improve with their
methods, and asks whether their result could be optimal under 
ETH (the Exponential-Time Hypothesis~\cite{ImpagliazzoPZ01}).

Using the above-mentioned LP-relaxation, we would get an algorithm only
for the case that the group is given in explicit form (i.e., not as an
oracle); in particular, we would have to limit ourselves to groups of
polynomial size. However, many useful cases of GFVS (including the
reductions from \textsc{Feedback Vertex Set} and 
\textsc{Subset Feedback Vertex Set}) use exponential-sized groups, and 
hence require the oracle form. To cover this case, we provide an
alternative LP-relaxation of the problem, which has an exponential number
of constraints, but which can be solved using a separation oracle. 
This implies an $O^*(4^k)$-time FPT algorithm for \textsc{Group Feedback
  Vertex Set} with group given via oracle access, providing the first
single-exponential FPT algorithms for GFVS and for \textsc{Subset Feedback
  Vertex Set}, hence answering the questions of Cygan et
al.~\cite{CyganPP12}. The new running times are optimal under ETH.

\subsubsection{Linear-time FPT algorithms}
As we have described above, the LP-branching based on discrete relaxations is a promising approach to establish FPT
algorithms and to reduce $f(k)$ part of the running time.
However, its $\mathrm{poly}(n)$ part is not so small since it relies on linear programming to solve the relaxations.
Reducing the $\mathrm{poly}(n)$ part is also an important task in FPT algorithms.
Especially, there have been many researches on FPT algorithms whose $\mathrm{poly}(n)$ part is only linear (linear-time
FPT), e.g., \textsc{Tree-Width}~\cite{Bodlaender96} and \textsc{Crossing
Number}~\cite{kawarabayashi07linear}.
Very recently, linear-time FPT algorithms for \textsc{Almost 2-SAT} have been developed
independently by Ramanujan and Saurabh~\cite{Ramanujan14}, and Iwata, Oka and Yoshida~\cite{Iwata14}.
The idea of the algorithm by Iwata~et~al. is to reduce the computation of LP relaxation to a minimum cut,
and actually, this approach works for solving several of our relaxation problems.
This approach generalise the linear-time FPT algorithm for \textsc{Almost 2-SAT} and gives the first linear-time FPT
algorithm for edge-deletion \textsc{Unique Label Cover} that runs in $O(|\Sigma|^{2p} m)$ time.
Thus the LP-branching based on discrete relaxations has a potential to reduce both $f(k)$ and
$\mathrm{poly}(n)$ simultaneously.

\subsection{Related work}
Hochbaum~\cite{Hochbaum02} gave a general framework for half-integral
relaxations of certain optimisation problems (as discussed above), via a
form of integer program called IP2 (which in turn is solved via relaxation
to a polynomial-time solvable problem class called \emph{monotone} IP2). 
%Her results cover a class of integer programs she calls IP2, with two or (in a restricted way) three
%variables per inequality (roughly corresponding to \emph{crisp} and
%\emph{soft} constraints in a VCSP). Hochbaum shows that so-called
%\emph{monotone} IP2 programs can be solved in polynomial time, and that
%half-integral relaxed solutions for general IP2 programs can be found via
%reduction to the monotone case. Furthermore, she gives a list of results
%covered by this framework. 
Without going into too much technical detail, we note that monotone IP2s
are covered in a VCSP framework by problems submodular \emph{on a
  chain}~\cite{JonssonKT11,IwataFF01,Schrijver00}, and that the
Boolean-domain case of IP2 reduces directly to \textsc{Almost 2-SAT},
a.k.a. \textsc{2-CNF Deletion}. However, we have not reconstructed a
direct VCSP interpretation of the full case of half-integral IP2. 
Hochbaum~\cite{Hochbaum02} asks in her paper whether the problems of
\textsc{Node Multiway Cut} and \textsc{Multicut on Trees} can be brought
into her framework; the problem of \textsc{Multicut on Trees} remains open
to us. 

Kolmogorov~\cite{Kolmogorov12DAM} gave close connections between
functions with half-integral minima and \emph{bisubmodular} functions, in
particular showing that bisubmodular functions correspond (in a certain
sense) to a class of (continuous-domain) functions referred to as
\emph{totally half-integral}. See Section~\ref{section:bisubmod} for more
details.  

%A different generalisation of \textsc{Multiway Cut} is the problem of
%\textsc{Submodular Multiway Partition}; see~\cite{ChekuriE11,EneVW13}. 

%There are several cases of LP-based approximations, even 2-approximations,
%which go beyond the scope of half-integrality, such as the method of
%iterative rounding~\cite{Jain01,FleischerJW06} and the primal-dual-based
%2-approximations for \textsc{Feedback Vertex Set}~\cite{ChudakGHW98});
%these methods are not known to imply any FPT connections. (Although the
%current work does imply an LP-based algorithm for \textsc{Feedback Vertex
%  Set}, there is no apparent connection between this algorithm and the
%above LP-formulations; in fact, a direct IP-formulation of our algorithm
%would have an exponential number of both variables and constraints.)

Submodular and bisubmodular functions also occur as rank functions of,
respectively, matroids~\cite{OxleyBook} and
delta-matroids~\cite{Bouchet95}; there are also connections to polytope
theory (e.g.,~\cite{BouchetC95}). 
%\todo{Something about ``covers matchings''.}
%Delta-matroids have also been generalised into a notion of \emph{jump
%  systems} in arbitrary lattices, however, less is known here? In
%particular there is no ``rank'' notion? Don't say this part?
Similar, but less well-explored connections exist for $k$-submodular
functions; see the theory of
multi-matroids~\cite{Bouchet97,Bouchet98,Bouchet01,BouchetIV}, 
and the polytope connection given by Huber and Kolmogorov~\cite{HuberK12}.

Group-labelled graphs (as in \textsc{Group Feedback Vertex Set}) and
bijection-labelled graphs (as in \textsc{Unique Label Cover}) have been
explored from a graph-theory perspective, in particular with respect to
path-packing; see~\cite{ChudnovskyCG08,ChudnovskyGGGLS06,KawarabayashiW06}
and~\cite{Pap07,Pap08}.

\section{Preliminaries}
\label{section:prel}

\subsection{Valued CSPs}
Let $D$ be a fixed, finite domain. A \emph{cost function} on $D$ (of arity
$r$) is a function $f: D^r \rightarrow \R$. A \emph{valued constraint} is an
application $f(v_1,\ldots,v_r)$ of a cost function $f: D^r \rightarrow \R$
to a tuple of variables $(v_1, \ldots, v_r)$. For simplicity, we disallow
repeated variables in constraints; this will make no difference for our
results but will simplify some notation.  A \emph{valued CSP instance}
(VCSP instance) is defined by a set $V$ of variables and a list of valued
constraints $f_1(v_{1,1}, \ldots, v_{1,r_1}), \ldots, f_m(v_{m,1}, \ldots,
v_{m,r_m})$, where $v_{i,j} \in V$ for each $i, j$; given an assignment
$\phi: V \rightarrow D$ and a VCSP instance $I$, we define the \emph{total
  cost} of $\phi$ for $I$ as 
$
f_I(\phi) = \sum_{i=1}^m f_i(\phi(v_{i,1}), \ldots, \phi(v_{i,r_i})).
$
Given a (not necessarily finite) set $\cF$ of cost functions on domain
$D$, the \emph{valued CSP problem} VCSP$(\cF)$ is the following problem:
given a VCSP instance $I$ on variable set $V$, where every cost function
$f_i$ is contained in $\cF$, and a number $k$, find an assignment 
$\phi: V \rightarrow D$ such that $f_I(\phi) \leq k$. 
%(Naturally, as a decision problem one could also append a cost bound $k$,
%and ask whether $f_I(\phi)\leq k$ for any assignment $\phi$.)

A \emph{crisp} constraint is one which cannot be broken (e.g., of infinite
or prohibitive cost). Given a relation $R$, let the \emph{soft version of $R$}
denote the valued constraint such that $f(X)=0$ if $R(X)$ holds, and $f(X)=1$ otherwise. 
%Usually in CSPs, there is a distinction between soft constraints, as
%above, and \emph{crisp} constraints which take costs from $\{0,\infty\}$ 
%(e.g., a crisp constraint must be satisfied by every assignment).
%In the present work, this distinction will not be essential, since we will
%always be able to implement every ``crisp'' constraint we need by using a
%bound $k$ on the instance cost and creating $k+1$ copies of a soft
%constraint.  

%The complexity of the VCSP problem varies depending on the set $\cF$ of
%cost functions one allows. It follows from a sequence of
%papers~\cite{ThapperZ12FOCS,Kolmogorov12CoRR,ThapperZ13STOC} that for
%every set $\cF$ of cost functions, the problem VCSP$(\cF)$ either has a
%polynomial-time exact solution or it is NP-hard; furthermore, the positive
%cases can all be solved with the help of a canonical LP relaxation known
%as the \emph{basic LP} (see below). 
%
%The approximation complexity of VCSPs, up to the unique games conjecture,
%was previously given by Raghavendra~\cite{Raghavendra08}. 

We will be most interested in the class of \emph{$k$-submodular}
functions, defined as follows. Fix a domain $D=\{0,1,\ldots,k\}$, and let
$\sqcap, \sqcup$ be symmetric, idempotent operations such that
$0 \sqcap x = 0$ for any $x \in D$; $0 \sqcup x = x$ for any $x \in D$;
and $x \sqcap y = x \sqcup y = 0$ for any $x, y \in D\setminus \{0\}$ with
$x \neq y$. % (note that this defines the operations completely). 
A function $f: D^r\rightarrow \R$ is $k$-submodular if $f(X)+f(Y) \geq f(X
\sqcap Y)+f(X \sqcup Y)$ for all $X, Y \in D^r$. 
%$\langle \sqcap, \sqcup \rangle$ is a multimorphism for it. This is a
%special case of the class of functions submodular on a tree, hence
%VCSP$(\cF)$ is in P if $\cF$ are $k$-submodular functions. 
The case $k=2$ is referred to as \emph{bisubmodular functions}.

\subsection{The basic LP relaxation}
%e only sketch the basic LP;
%see~\cite{ThapperZ12FOCS} for a full version. (The version below uses
%fewer variables, but is easily seen to be equivalent.)
Since it is fundamental to our paper, let us explicitly define the LP
which lies behind all the above tractability results. Let $\cF$ be a
finite set of cost functions over a domain $D$, and let $I$ be an instance of VCSP$(\cF)$
%Let $I$ be a VCSP instance over a domain $D$, 
%of VCSP$(\cF)$ 
on variable set $V=\{v_1,\ldots,v_n\}$ and with valued
constraints $f_i(v_{i,1}, \ldots, v_{i,r_i})$, $1 \leq i \leq m$.
%, where $f_i \in \cF$ for each $i$ and $v_{i,j} \in V$ for every $i,j$. 
The \emph{basic LP relaxation} (BLP) of~$I$ is defined as follows.
(The definition given in~\cite{ThapperZ12FOCS} is
  slightly different, but can easily be verified to be equivalent to the
   formulation below for our case.)
Introduce variables $\mu_{v=d}$ for every $v \in V$ and $d \in D$,
and $\lambda_{f_i, \sigma}$ for every valued constraint $f_i$ in $I$ 
and every $\sigma \in D^{r_i}$. The (BLP) is defined as follows.
%The task is to minimise $\sum_{i=1}^m  \sum_{\sigma \in D^{r_i}} f_i(\sigma(1), \ldots, \sigma(r_i)) \cdot \lambda_{f_i, \sigma}$ 
%subject to $\sum_{d \in D} \mu_{v=d}=1$ for every $v \in V$,
%and ...
  \begin{alignat*}{2}
    \min               & \sum_{i=1}^m \sum_{\sigma \in D^{r_i}}
                   f_i(\sigma(1), \ldots, \sigma(r_i)) \cdot \lambda_{f_i, \sigma} \\
    \mathrm{s.t.}  &   \sum_{d \in D} \mu_{v=d} = 1 &\quad & \forall v \in V \\
                 & \sum_{\sigma \in D^{r_i}: \sigma(j)=d} \lambda_{f_i, \sigma} = \mu_{v=d} &\quad & 
                   \forall 1 \leq i \leq m, 1 \leq j \leq r_i, d \in D, v=v_{i,j} \\
                 & 0 \leq \lambda_{f_i,\sigma}, \mu_{v=d} \leq 1 
  \end{alignat*}
%The LP asks to minimise $\sum_{i, \sigma} f_i(\sigma) \cdot
%\lambda_{f_i,\sigma}$, subject to $\sum_d \mu_{v=d}=1$ for every $v \in V$
%and $\sum_{\sigma: \sigma(j)=d} \lambda_{f_i,\sigma}=\mu_{v_{i,j}=d}$ for every
%$i \in [m], j \in [r_i],d \in D$. 
Note that the size of the LP depends badly on function arity, e.g.,
we introduce $|D|^r$ variables for a single $r$-ary valued constraint $f$. 
However for every finite set of functions, as required above, this arity
is bounded and the LP is of polynomial size.
We will later in the paper define smaller, equivalent LP-relaxations for
particular problem classes. 

To reiterate, it is a consequence of~\cite{ThapperZ12FOCS} that if $\cF$
is a set of $k$-submodular functions, then the above LP solve VCSP($\cF$)
precisely. 

\subsection{Polymorphisms and fractional polymorphisms}
A key tool in the characterisation of CSP complexity is the
\emph{algebraic method}. For a domain $D$, an operation $h: D^t
\rightarrow D$, and a list of tuples $A_1, \ldots, A_t \in D^\ell$, 
define $h(A_1, \ldots, A_t) \in D^\ell$ as the result of applying $h$
\emph{column-wise} to the tuples, i.e., if $A(j)$ denotes the $j$-th
entry of a tuple $A$, we let
%\begin{align*}
%h(A_1, \ldots, A_t) = &(h(A_1(1), \ldots, A_t(1)), \\&\phantom{(}\ldots, 
%                     \\&\phantom{(}h(A_1(\ell), \ldots, A_t(\ell))).
%\end{align*}
$h(A_1, \ldots, A_t)$ be the tuple $T \in D^\ell$ such that
$T(i)=h(A_1(i), \ldots, A_t(i))$.
Given a relation $R \subseteq D^r$, a \emph{polymorphism} of $R$ is an
operation $h: D^t \rightarrow D$ such that for any tuples $A_1, \ldots,
A_t \in R$, we have $h(A_1,\ldots,A_t) \in R$ (i.e., the relation $R$ is
closed under the operation of applying $h$ column-wise on any set of $t$
tuples in $R$). For a set of relations $\Gamma$, we say that $\Gamma$ has
a polymorphism $h$ if $h$ is a polymorphism of every $R \in \Gamma$. 
It is known that the complexity of classical (feasibility) CSP$(\Gamma)$
is characterised by the set of polymorphisms of the allowed relation types
$\Gamma$, however, no complete dichotomy is known for this question. 

We will need only the following notion: A \emph{majority} polymorphism is
a polymorphism $h: D^3 \rightarrow D$ such that
$h(x,x,y)=h(x,y,x)=h(y,x,x)=x$ for any $x, y \in D$. 
It is known that for any set of relations $\Gamma$ with a majority
polymorphism, the solution set for any formula over $\Gamma$ can be
described using only \emph{binary} relations (derivable from $\Gamma$);
see~\cite{JeavonsCC98}.

For valued constraints, the notions must be expanded to
\emph{fractional} polymorphisms; see~\cite{ThapperZ12FOCS,ThapperZ13STOC}
for definitions, and for an exact characterisation of the VCSP dichotomy
results. For this paper, we will be content with a simpler notion. 
Let $f: D^r \rightarrow \R$ be a cost function. A \emph{binary
  multimorphism} of $f$ is a pair of operations 
$\langle h_1, h_2\rangle: D^2 \rightarrow D$ such that for any $A, B \in
D^r$, we have $f(A)+f(B) \geq f(h_1(A,B)) + f(h_2(A,B))$. 
Similarly to above, $\langle h_1,h_2 \rangle$ is a multimorphism of a set
$\cF$ of cost functions if it is a multimorphism of every $f \in \cF$. 
The prime example would be the \emph{submodular} functions, which are
defined on domain $D=\{0,1\}$ by the multimorphism $\langle \cap, \cup
\rangle$ (i.e., $f(X)+f(Y) \geq f(X \cap Y)+f(X \cup Y)$); it is well known
that submodular functions can be minimised efficiently
(e.g.,~\cite{IwataFF01,Schrijver00}). 
Other examples of function classes $\cF$ which imply that VCSP$(\cF)$ is
tractable include (among other cases) functions submodular on an arbitrary
lattice, defined as having the multimorphism $\langle \lor, \land \rangle$,
and functions (weakly or strongly) submodular on a tree;
see~\cite{ThapperZ12FOCS} for details.

\section{Discrete Relaxations and FPT Branching}
\label{section:discrete-relaxation}

We now describe our approach more precisely. 

\begin{definition}
Let~$f: D^r \rightarrow \R$ be a finite-valued cost function. 
A \emph{discrete relaxation} of~$f$ on domain~$D' \supset D$ is a
function~$f': (D')^r \rightarrow \R$, such that 
(i)~$\min_{\bar x \in (D')^r} f'(\bar x)=\min_{\bar x \in D^r} f(\bar x)$, and 
(ii)~$f(\bar x)=f'(\bar x)$ for every~$\bar x \in D^r$. A discrete relaxation
of a set of cost functions $\cF=\{f_1, \ldots, f_t\}$ is a set of cost
functions~$\cF'=\{f_1',\ldots,f_t'\}$ on a domain~$D' \supset D$,
such that~$f_i'$ is a discrete relaxation of~$f_i$ for each~$i \in [t]$.
Finally, given an instance~$I$ of VCSP$(\cF)$, the \emph{relaxed instance}~$I'$
of VCSP$(\cF')$ is created by replacing every cost function~$f_i$ in~$I$ by 
its corresponding relaxation~$f_i'$.
The \emph{(additive) relaxation gap} of $I$ is OPT$(I)$-OPT$(I')$.
\end{definition}

Note that we can have $\mathrm{OPT}(I) > \mathrm{OPT}(I')$ despite every
individual cost function $f_i$ having an identical minimum (e.g., if
setting $v=d'$ for every variable $v$ minimises every constraint, for some
$d' \in D' \setminus D$). 
If~$\cF$ is integer-valued, let the \emph{scaling factor} of the
relaxation~$\cF'$ be the smallest rational~$c$ such that~$c \cdot f_i'$ is
integral for every~$f_i'\in\cF'$. In this case, we say that $\cF'$ is a
\emph{$c$-relaxation} of $\cF$ (note that this does not necessarily imply
that $I'$ is an approximation). 

\begin{definition}
Let $\cF$ be a set of cost functions on a domain~$D$, with a discrete
relaxation $\cF'$ on domain $D'$. We refer to the values of~$D$ as the 
\emph{original} values, and~$D' \setminus D$ as the \emph{relaxed} values.
In an assignment~$\phi: V \rightarrow D'$, we say
that a variable~$v \in V$ is \emph{integral in $\phi$} if~$\phi(v) \in D$;
otherwise,~$v$ is \emph{relaxed in~$\phi$}.  
An assignment~$\phi$ is integral if it uses only original values, i.e., if every
variable~$v \in V$ is integral in~$\phi$. 
Borrowing a term from Kolmogorov~\cite{Kolmogorov12DAM}, we say that the relaxation is
\emph{persistent} if, for any optimal assignment~$\phi^*$ of a relaxed instance~$I'$, 
there is an optimal \emph{integral} assignment~$\phi$ that agrees with~$\phi^*$ on the
latter's integral values (i.e., if~$\phi^*(x)$ is integral, then~$\phi(x)=\phi^*(x)$). 
\end{definition}

As a slight technical point, note that persistence is a function of the
division of the domain $D'$ into integral and relaxed parts, and does not
explicitly require a reference to an original function on a domain $D$
being relaxed. In our main case, we will deal with functions on a domain of
$D=\{1,\ldots,k\}$, which have relaxations on a domain $D'=\{0,\ldots,k\}$
which are \emph{$k$-submodular}. Thus, we will have a single relaxed
domain value of~$0$.  

To illustrate the notions, we show the application to \textsc{Vertex
  Cover}. 
%\label{ex:vc}
Consider the Boolean domain~$D=\{0,1\}$. Let~$f_{\lor}$ be defined by
$f_{\lor}(0,0)=1$, and $f_{\lor}(x,y)=0$ otherwise (i.e., $f_{\lor}$ is
the soft version of the relation $(x \lor y)$), and let~$f_0(x)=x$
(corresponding to the soft version of requiring $x=0$). 
Then VCSP$(f_{\lor},f_0)$ is NP-hard, as it encodes \textsc{Vertex Cover}
when $f_{\lor}$ is treated as a crisp constraint.
On the other hand, let~$D'=\{0,\half,1\}$, and define the relaxations
$f'_{\lor}(x,y)=\max(0, 1-x-y)$ and $f_0'(x)=x$.
Then this is a discrete relaxation of the original problem, which
furthermore is a persistent 2-relaxation and can be solved in polynomial
time, as it corresponds to the classical LP-relaxation of \textsc{Vertex
  Cover} (see Nemhauser and Trotter~\cite{NemhauserT75}). 
Furthermore, the relaxed functions are bisubmodular if $D'$ is renamed as
$(0,\half,1) \mapsto (1,0,2)$.
%\end{example}

This example also roughly illustrates the connections between tractable
discrete relaxations and half-integrality. From~\cite{ThapperZ12FOCS} we
have that for every tractable set of cost functions $\cF$, and every
instance $I$ of VCSP$(\cF)$, the optimum of the basic LP relaxation (BLP)
coincides with OPT$(I)$. Since the results of~\cite{ThapperZ12FOCS} 
support weighted functions (e.g., an input of $w_i \cdot f_i(\cdot)$ 
rather than just $f_i(\cdot)$), and since such weights only occur in the
cost function of the LP, it must be that every vertex of the LP is
integral, i.e., that (BLP) is an integral LP. Now, rather than a
half-integral LP, this is an integral LP on a different, larger
set of variables, however, in the cases considered in this paper
(bisubmodular and $k$-submodular functions), we will see that such a larger
LP can (at least in specific cases) be mapped down to a half-integral LP
on the original variable set. 

%
%Note that persistence is a property of~$\Gamma'$ and of the choice~$D
%\subseteq D'$ of integral values, and hence does not require every function
%in~$\Gamma'$ to be a valid relaxation of a function on~$D$ (i.e., we may allow
%5functions which attain their minima only outside of~$D$).  
%
%
Persistent relaxations are key to providing FPT algorithms, as the following shows. 

%The following
%can be shown by adapting standard LP-branching algorithms. 
%(Note that the result works without explicit access to a satisfying assignment.)

\begin{lemma}
\label{lm:relaxfpt}
Let~$\cF$ be a set of integer-valued cost functions on~$D$, and let~$\cF'$
be a persistent $c$-relaxation of~$\cF$ on domain~$D'$, which includes all
hard constants from~$D$ (i.e., for each $d \in D$ there is either a crisp constraint $(v=d)$ 
or a valued constraint $f_d(v)$ for which $v=d$ is the unique minimum). 
Given black-box access to a solver
for VCSP$(\cF')$, we can solve an instance~$I$ of VCSP$(\cF)$
using~$\Oh^*(|D|^{ck})$ calls to the black-box solver and polynomial
additional work, where~$k=\mathrm{OPT}(I)-\mathrm{OPT}(I')$ is the additive
relaxation gap.
\end{lemma}
\begin{proof}
Let $I$ be the input instance, and $I'$ the relaxed instance. 
Let $x^*=\mathrm{OPT}(I')$, and let $k$ be a (guessed) bound on the
relaxation gap. Pick an arbitrary variable $v \in V$, and attempt to
enforce $(v=d)$ for every $d \in D$ in turn (e.g., by a sufficient number
of copies of the valued constraint $f_d(v)$). If there is a value $d \in D$
such that enforcing $(v=d)$ fails to increase the optimal cost of $I'$, 
then add the enforcing of $(v=d)$ to $I'$, and proceed with another
variable (if possible); this is legal since the approximation is persistent. 
If every variable $v \in V$ is part of a forced assignment, then we have
an integral solution, which must be optimal since $I'$ is a relaxation. 
In the remaining case, every enforced assignment $v=d$ raises the cost of
$I'$. In this case, we simply recurse into $|D|$ directions according to all
possible assignments; in each branch, the gap parameter $k$ has decreased
by at least $1/c$. Halt a recursion if the gap parameter reaches $0$.
We get a tree with branching factor $|D|$ and depth at most $ck$, implying
the result. 
\maybeqed{} \end{proof}

For some problems, with some extra work, we can remove the factor~$|D|$ from
the base of the above running time; however, this is not possible in
general unless FPT=W[1] (see Section~\ref{section:ksubmod}). 

In the rest of this section, we focus on the case when the relaxation is a
bisubmodular function, and show how this case explains and extends certain
results of half-integrality from the literature; in the rest of the paper,
we focus on cases of $k$-submodular functions, and new results which follow from
those. 

\subsection{Case study: Submodular and bisubmodular functions}
\label{section:bisubmod}

As mentioned in Section~\ref{section:prel}, a \emph{bisubmodular function} is
defined as a function $f: \{0,1,2\}^r \rightarrow \R$ which satisfies a
certain multimorphism equation ($f(A)+f(B) \geq f(A \sqcup B)
+ f(A \sqcap B)$ for all $A, B \in \{0,1,2\}^r$). However, a more fitting
interpretation may be to remap the domain to $D'=\{0,\half,1\}$, whereupon
the operations $\sqcap, \sqcup$ can be defined as $\{(x \sqcap y), 
(x \sqcup y)\}= \{\lceil (x+y)\rceil/2,\lfloor (x+y) \rfloor/2\}$,
where $\sqcup$ rounds away from \half and $\sqcap$ towards \half.
%\ynote{Is this definition correct? When $x=y=1$, $\frac{1}{2}\lceil (x+y)/2 \rceil=\frac{1}{2}$, but $x \sqcap y$ should be $1$}
In this setting, we would interpret \half as a \emph{relaxed} value, and
$0$ and $1$ as \emph{integral}. 
Kolmogorov~\cite{Kolmogorov12DAM} showed that with this domain split,
bisubmodular functions are persistent. 
Furthermore, bisubmodular functions can be efficiently minimised even in
a value oracle model~\cite{FujishigeI05}.

Thus, by applying Lemma~\ref{lm:relaxfpt}, we get that for any class of
integer-valued cost functions $\cF$ on a domain $\{0,1\}^n$, with a
bisubmodular discrete $c$-relaxation, the problem VCSP$(\cF)$ is FPT with
a running time of $O^*(2^{ck})$, parameterized by the relaxation gap $k$
(where we will find that the factor $c=2$ suffices for all our cases).
%Let us investigate what this class of functions contains.
%
%For brevity, if a function $f$ on domain $\{0,1\}$ has a discrete
%relaxation $f'$ which is bisubmodular on $\{0,\half,1\}$, we simply say that
%$f$ has a bisubmodular relaxation. 
We re-derive some known FPT consequences.

\begin{corollary}[\cite{NarayanaswamyRRS12}]
\label{cor:vc-a2sat}
\textsc{Vertex Cover Above LP}, \textsc{Min Ones 2-CNF Above LP}, and
\textsc{Almost 2-SAT} are all FPT with a running time of~$\Oh^*(4^k)$.%
%\footnote{Note that the results of~\cite{NarayanaswamyRRS12}
%  and~\cite{LokshtanovNRRS12CoRR} actually derive running times~$\Oh^*(c^k)$,
%  for some constants $c<4$.}
\end{corollary}
\begin{proof}%[Proof of Corollary~\ref{cor:vc-a2sat}] %%%%%MARK!
For \textsc{Vertex Cover}, we simply repeat the construction in the example.
Let~$D'=\{0,\half,1\}$ as above, and define~$f_{\lor}(x,y)=\max(0,1-x-y)$ and~$f_0(x)=x$.
It can be verified that~$f_\lor$ and~$f_0$ are both bisubmodular
functions; by always using~$f_\lor$ at a weight of at least~$2n$, we may
emulate a crisp (unbreakable) or-constraint. Furthermore, we have
assignments $(x=0)$ and $(x=1)$: in the former case via $2n$ copies of
$f_0(x)$; in the latter, via $2n$ copies of $f_{\lor}(x,z_0)$ where $z_0$
is some new variable forced to take value $0$. 
Thus Lemma~\ref{lm:relaxfpt} applies. 

To capture \textsc{Min Ones 2-CNF} and \textsc{Almost 2-SAT}, we observe
that the further functions~$f_{\land}(x,y) = \max(0,x+y-1)$
and~$f_{\rightarrow}(x,y)=\max(0, x-y)$ are also bisubmodular, and
furthermore valid relaxations of the corresponding soft versions of 2-clauses.
\maybeqed{} \end{proof}

By the existence of a value oracle minimiser, we can extend to showing that 
the problem \textsc{Bisubmodular Cost 2-SAT}, defined below,
is FPT with a running time of $O^*(2^k)$ %\ynote{Should this be $O^*(2^k)?$}.  
(Since bisubmodular functions are closed under adding or subtracting a constant,
we may assume that~$f$ attains the value zero on~$\{0,\half,1\}^V$, 
hence the total cost parameter $k$ has the same power as a relaxation gap parameter would.)

\parameterizedproblem{Bisubmodular Cost 2-SAT}%
{2-CNF $F$ on variable set~$V$, 
 non-negative bisubmodular function~$f: \{0,\half,1\}^V \rightarrow \Z$ (with
 black box access), 
 integer~$k$.}%
{$k$}
{Is there a satisfying assignment~$\phi: V \rightarrow \{0,1\}$ for~$F$ with~$f(\phi)\leq k$, where~$f(\phi)=f(\phi(v_1),\ldots,\phi(v_n))$ is
  the value of $f$ under $\phi$?}

%Again, this is FPT, with running time $O^*(4^k)$. %As the proof shows, the result also covers \textsc{Almost 2-SAT} directly
%(rather than via the translation from \textsc{Vertex Cover Above Matching}). 

\begin{corollary}
\label{cor:bisub:2cnf}
\textsc{Bisubmodular Cost 2-SAT} is FPT, with a running time of~$\Oh^*(2^k)$.
\textsc{Submodular Cost 2-SAT} under the same parameter is FPT with 
a running time of~$\Oh^*(4^k)$, even for non-monotone submodular cost
functions. 
\end{corollary}
\begin{proof}%[Proof of Corollary~\ref{cor:bisub:2cnf}]%%%%MARK!
First, we may enforce the crisp 2-CNF formula $F$, as previously noted,
by creating large-weight finite-valued constraints for the 2-clauses. 

For bisubmodular cost functions, the corollary follows in a straight-forward
manner.  Let~$M$ be a value large enough to dominate the cost of~$f$ (such
a value can be found, if nothing else, by repeating the below with
gradually higher values of $M$), and construct a new
bisubmodular cost function~$f'=f+\sum_{C \in   F} M\cdot f(C)$,
where~$f(C)$ for a 2-clause~$C$ is the corresponding function defined in
Corollary~\ref{cor:vc-a2sat}. Then any minimizer of~$f'$ must satisfy the
LP-relaxation of~$F$. Since~$f$ is already integer-valued, our ``scaling
factor'' is 1, and the running time follows. 

For submodular functions, we observe that the Lov\'asz extension, evaluated
on $\{0,\half,1\}^V$, is a bisubmodular function, and thus a bisubmodular
relaxation with scaling factor 2. 
To be explicit, consider some $A \in \{0,\half,1\}^V$, 
decomposed as $A=A_1 + \frac 1 2 A_{\half}$ for $A_1, A_{\half} \subseteq V$,
and write~$A_h=A_1\cup A_{\half}$; proceed similarly for a second point~$B$.
By the definition of the Lov\'asz extension and submodularity we have
\begin{align*}
2\hat f(A) + 2\hat f(B) 
   & =    f(A_1) + f(A_h) + f(B_1) + f(B_h) \\
   & \geq f(A_1 \cap B_1) + f(A_1 \cup B_1) +  f(A_h \cap B_h) + f(A_h \cup B_h) \\
   & \geq f(A_1 \cap B_1) + f(A_h \cup B_h) +  f((A_1 \cup B_1) \cap (A_h \cap B_h)) 
          \\ & \quad +  
      f((A_1 \cup B_1) \cup (A_h \cap B_h)),
\end{align*}
where it can be verified that the last four terms are exactly the same as would be
produced by applying the bisubmodular operators~$\sqcap, \sqcup$ on~$A, B$ directly and
evaluating the result. 
\maybeqed{} \end{proof}

The particular case of \textsc{Submodular Vertex Cover} was previously
shown to have a half-integral relaxation~\cite{IwataN09}; the above shows
that this problem is also FPT. 

Although it is difficult to get a good handle on the expressive power of
bisubmodular functions in general, let us mention that beyond submodular
functions, the class also covers \emph{twistings} $f(S \Delta X)$ of
submodular functions $f(X)$ (for some fixed $S \subseteq V$), sums of such
twistings, and (perhaps more generally) rank functions of
delta-matroids~\cite{Bouchet95}. 

In the appendix, we make a note observing that the use of a 2-CNF formula
$F$ precisely captures the ``crisp expressive power'' of bisubmodular
relaxations (in the same way as a ring family for submodular functions; see Schrijver~\cite{SchrijverBook}). 

\subsection{Edge- versus vertex-deletion problems}
\label{section:vdel}
%\todo{Musings on edge- versus constraint-deletion costs in source.
%      May be reinstated.}%%%%MARK!
Finally, we note that  the above discussion is generally described on an \emph{edge} or
%Finally, note that
\emph{constraint deletion} level (e.g., a typical pre-relaxation cost
function is a function $f: \{0,1\}^r \rightarrow \{0,1\}$ encoding
the soft version of some relation $R \subseteq \{0,1\}^r$). 
In several problems (in particular in the following sections), one may
wish to also express the \emph{vertex} or \emph{variable deletion}
version. This can be done as follows. For a variable $v$, occurring in $d$
different constraints, we introduce a separate variable $v(1), \ldots,
v(d)$ for each occurrence, we give each individual constraint on these new
variables high enough weight that it will be treated as crisp, and we
impose a valued constraint $(v(1)=\ldots=v(d))$ (a \emph{soft wide equality}),
which takes value $0$ if all occurrences of $v$ are identical and value $1$ otherwise. 
These constraints would effectively encode whether a variable $v$ has been
deleted (with constraint weight 1, e.g., every occurrence $v(i)$ of $v$
can take whatever value it needs to satisfy its constraint) or not. 
Note that these soft wide equalities are defined on the original domain,
and hence need to admit an appropriate discrete relaxation; for the case
of $k$-submodular relaxations, this is possible. 

A bigger problem is that these constraints have unbounded arity. 
For bisubmodular functions, this is acceptable, both since 
we may use a value oracle model, and since it has an implementation
as a 2-CNF formula with additional variables, e.g., 
$(v(1) \rightarrow y) \land \ldots \land (v(d) \rightarrow y) \land (y
\rightarrow z) \land (z \rightarrow v(1)) \land \ldots \land (z \rightarrow
v(d))$. Unfortunately, neither of these options is available for
$k$-submodular functions; we will instead need to construct a different
LP. 

\section{On the power of $k$-submodular relaxations}
\label{section:ksubmod}

We now investigate the power of $k$-submodular functions for discrete
relaxation, that is, we investigate the class of cost functions $f$ on 
a domain $D=\{1,\ldots,k\}$ which have discrete relaxations $f'$ on
the domain $D'=\{0,\ldots,k\}$ such that $f'$ is a $k$-submodular function.
We will find that this covers both some well-known half-integrality
results (e.g., the \textsc{Multiway Cut} problem~\cite{GargVY04}) and
several new results that one might not have suspected (e.g., half-integral
relaxations of \textsc{Group Feedback Vertex Set} and \textsc{Unique Label
  Cover}). 

We begin with establishing the basic essential properties. 

\begin{lemma}
\label{lemma:ksubmodpers}
The class of~$k$-submodular functions, on domain~$D'=\{0,\ldots,k\}$, is
persistent with respect to a choice of integral domain~$D=\{1,\ldots,k\}$.
Furthermore, it contains all hard constants from~$D$; specifically, for 
each $d \in D$ there is a unary valued constraint $f_d(v)$ which has $v=d$ as a
unique minimum. 
\end{lemma}
\begin{proof}
For persistence, consider the following derivation. Let~$f$ be a cost
function,~$X^*$ a relaxed optimum, and~$X$ an integral optimum.
\begin{align*}
f(X) + 2f(X^*) &\geq f(X \sqcap X^*) + f(X \sqcup X^*) + f(X^*) \\
               &\geq f(X \sqcap X^*) + f((X \sqcup X^*) \sqcap X^*) + 
                %\\ & \quad +  
f((X \sqcup X^*) \sqcup X^*) \\
               &\geq 2f(X^*) + f((X \sqcup X^*) \sqcup X^*),
\end{align*}
where the first two lines are due to application of~$k$-submodularity
equality, and the last line is since~$f(X^*)$ is a relaxed optimum. 
Thus~$f(X) \geq f((X \sqcup X^*) \sqcup X^*)$ for any integral optimum~$X$ and
relaxed optimum~$X^*$. Observe now that the latter operation preserves all
coordinates from~$X$ where~$X^*$ takes value zero, and replaces all other
coordinates (where~$X^*$ is integral) by the value from~$X^*$. Thus the
right-hand-side of this equation is an integral optimum which agrees
with~$X^*$ on the integral coordinates of the latter.

For the last part, we define $f_d(v)$ such that $f_d(d)=0$;~$f_d(0)=\half$; and~$f_d(d')=1$ for 
any~$d' \in D, d' \neq d$. 
\maybeqed{} \end{proof}

\begin{corollary}
\label{cor:ksubmod:direct}
For any set $\cF$ of bounded-arity functions on a domain $\{1,\ldots,k\}$,
with a known $k$-submodular $c$-relaxation $\cF'$, the problem VCSP$(\cF)$
is FPT with a running time of $O^*(k^{cp})$, where $p$ is the relaxation
gap.
\end{corollary}

%\todo{Reinstate commented-out parts?}%%%MARK!
The restriction of arity is due to the size of the Basic LP relaxation.
Unfortunately, as mentioned in Section~\ref{section:vdel},
this is a significant restriction if one wants to support vertex deletion problems.
%, as there are cases (for vertex deletion
%problems) where one wants constraints of arity of order $n$.
% (in which case the size of the LP would already be not only super-polynomial, but super-FPT). 

In the rest of this section, we first establish a basic collection of
functions with $k$-submodular relaxations (and make a note on the
structure of $k$-submodular optima), then provide an alternate
LP-relaxation for this particular set of functions, to get
around the problem of arity. Finally, we make a note on the parameterized
complexity of the \textsc{Unique Label Cover} problem. 
We then  study the \textsc{Group Feedback Vertex Set} problem in Section~\ref{section:gfvs}.

\subsection{Basic $k$-submodular functions}
Now, let us establish some basic $k$-submodular relaxations.
%The \emph{soft version} of a constraint $R\subseteq D^r$ is a
%function $f: D^r \rightarrow \{0,1\}$ such that $f(A)=0$ if and only if $A
%\in R$. 

\begin{lemma}
\label{lemma:usefulksubmod}
The following cost functions on a domain~$D=\{1,\ldots,k\}$
have~$k$-submodular relaxations.
We let $x, y$ denote variables and $d, d'$  domain values.
\begin{enumerate}
\item
%(i)
Any unary function;
\item
%(ii) 
the soft version of a constraint~$(x = \pi(y))$, for any permutation~$\pi$ on~$D$;
\item 
%(iii) 
the soft version of a constraint~$(x=d \lor y=d')$ for~$d, d' \in D$;
\item 
%(iv) 
the soft version of the constraint~$(x_1=\ldots=x_r)$. 
\end{enumerate}
The scaling factor in all cases is 2. 
\end{lemma}
\begin{proof}
We supply only the relaxations here; the proof that each relaxation is actually $k$-submodular
is straight-forward case analysis, deferred to the appendix. 

\emph{1.} For the first case, we may simply relax by stating~$f'(0)=\min_{d \in D} f'(d)$. 
We may also use a slightly stronger version, as follows. 
Put~$d_1=\arg \min_{d \in D} f(d)$, and~$d_2=\arg \min_{d \in D: d \neq d_1} f(d)$. 
Then we may use
\[
f'(0) = \frac{f(d_1) + f(d_2)}{2}.
\]
In particular, this covers ``hard constants'' on~$D$.

\emph{2.} For the second case, define a relaxation~$f$ such that~$f(0,0)=0$
and~$f(a,0)=f(0,a)=\half$ if~$a \neq 0$. 

\emph{3.} For the third case, with specified domain elements $d, d' \in D$,
let~$f_{d,d'}$ on~$D'$ be the extension of the original valued constraint to~$D'$ as 
follows:~$f_{d,d'}(d,0)=f_{d,d'}(0,d')=f_{d,d'}(0,0)=0$,
and~$f_{d,d'}(0,d'')=f_{d,d'}(d'',0)=\half$ for all remaining cases. 

\emph{4.} For the soft wide equality function, define a relaxation as follows. If a tuple contains distinct
integral values, the cost is 1; if a tuple contains some integral value and
the value 0, the cost is \half; if the tuple is constant, the cost is 0.

This completes the cases.
\maybeqed{} \end{proof}

Via Corollary~\ref{cor:ksubmod:direct}, this implies that VCSP$(\cF)$ is
FPT when $\cF$ contains bounded-arity versions of the above cost
functions. The constraint $(x=d \lor y=d')$ is included mostly for
completeness (see below, regarding the solution structure), although it
does allow for a generalisation of how \textsc{Almost 2-SAT} could be
encoded into a bisubmodular cost function. The case of bijection
constraints is more interesting, as it allows for a direct encoding of
\textsc{Unique Label Cover} (see Section~\ref{section:ulc}) and problems
related to \textsc{Group Feedback Edge/Vertex Set} problems (see
Section~\ref{section:gfvs}).  Finally, the soft wide equality constraints
imply that we could \emph{in principle} handle vertex-deletion, if we had
a better underlying solver than the Basic LP; this is tackled in
Section~\ref{section:betterlp}.

As for bisubmodular functions, we show that the cases of Lemma~\ref{lemma:usefulksubmod} 
are sufficient to capture the crisp expressive power of functions with $k$-submodular
relaxations; the proof is in the appendix. 
Interestingly, this coincides with the language of so-called \emph{0/1/all constraints}
of Cooper et al.~\cite{CooperCJ94}, who showed this to be the unique maximal tractable 
CSP language closed under all permutations of the domain (see~\cite{CooperCJ94}).

\begin{lemma}
\label{lemma:ksubstructure}
Let $f$ be a $k$-submodular function on $D^n$, and let $P \subseteq D^n$
be the set of points $X$ that minimise $f(X)$. Let $P_{\mathrm{int}}=P \cap
\{1,\ldots,k\}^n$. Then $P_{\mathrm{int}}$ can be described as the set of
solutions to a formula over arbitrary unary constraints and constraints
$(x=a \lor y=b)$ and $(x=\pi(y))$ (defined as in Lemma~\ref{lemma:usefulksubmod}). 
\end{lemma}

Note that this does capture the whole structure of minima of $k$-submodular functions,
due to the special way in which we treat the element $0$. Furthermore, and more
strongly, this does not limit the expressive power of $k$-submodular functions in 
general, as it focuses purely on the structure of minima. 
(See discussion in appendix.)

For our purposes, it also implies that if $R$ is a relation on domain $D$ whose soft version 
has a $k$-submodular relaxation, then $R$ can be expressed as a conjunction over the constraints above.
However, we do not know whether the \emph{soft version} of $R$ can in this case necessarily be 
\emph{implemented} as such a formula (taking costs $0$ and $1$ only). 

\subsection{A half-integral LP formulation}
\label{section:betterlp}

We now proceed to give an alternate half-integral LP-formulation for the
$k$-submodular relaxations given in Lemma~\ref{lemma:usefulksubmod}. 
The construction is somewhat modelled after the half-integral LP for
\textsc{Node Multiway Cut} given by Garg~et~al.~\cite{GargVY04}.
Let the input be an instance $I$ of VCSP$(\cF)$ with $m$ constraints, 
where $\cF$ is the set of cost functions given in Lemma~\ref{lemma:usefulksubmod}. 
Let the variable set of the VCSP be $V=\{v_1,\ldots,v_n\}$. We split every variable 
$v_i \in V$ in the CSP into $k$ variables $v_{i,d}$, 
one for every $d \in [k]:= \{1, \ldots, k\}$. 
Further, for every constraint $f_j$ of $I$, we introduce a variable $z_j$
to take care of the cost of $f_j$. Define a set $A$ to contain all pairs $(i,d)$
such that an assignment $(v_i=d)$ is to be enforced. 
The \emph{framework} constraints of the LP are as follows.
\begin{alignat*}{2}
\min\quad & \sum_j z_j \\
\mathrm{s.t.}\, & v_{i,a}+v_{i,b} \leq 1 &\,& \forall i \in [n], a,b \in [k], a \neq b \\
              & v_{i,d}=1             &\,& \forall (i,d) \in A \\
              & v_{i,d}, z_j \geq 0   & & \forall i \in[n], d \in [k], j \in [m]
\end{alignat*}
%The objective is to minimise $(\sum_j z_j)$, subject to constraints as follows.
%The \emph{framework} constraints state $(v_{i,a}+v_{i,b}\leq 1)$ for every
%$i \in [n], a,b \in [k], a\neq b$; $(v_{i,d}=1)$ for every $(i,d)\in A$; 
%and $v_{i,d}, z_j\geq 0$ for every $i \in [n], j \in [m], d \in [k]$. 
Further constraints bound the value of $z_j$; throughout, we use the relaxation
functions of Lemma~\ref{lemma:usefulksubmod}.
%The ``framework'' constraints of the LP are as follows; further
%constraints to fix the values of $z_j$ are given below, on case-by-case
%basis. Let $[k]=\{1,\ldots,k\}$. 
%We treat assignments $(v_i=d)$ separately; let $A$ contain all pairs
%$(i,d)$ such that an unbreakable assignment $(v_i=d)$ is desired. 
%  \begin{alignat*}{2}
%    \min               & \sum_{j=1}^m z_j \\
%    \mathrm{s.t.}  &  v_{i,d}+v_{i,d'} \leq 1 &,\ & 
%             \forall v_i \in V, d,d' \in [k], d \neq d' \\
%             & v_{i,d}=1 &,\ & (i,d) \in A \\
% & z_j 
%             v_{i,d}, z_j \geq 0 \forall i, d, j
%  \end{alignat*}
If $f_j(v_i)$ is a unary cost function, 
%in the input, extended to value $0$ as in Lemma~\ref{lemma:usefulksubmod},
%i.e., by 
let $f_j(0):=(f_j(d_1)+f_j(d_2))/2$, where
$d_1 = \arg \min_{x \in [k]} f_j(x)$ and 
$d_2= \arg \min_{x \in [k], x \neq d_1} f_j(x)$. %, as in Lemma~\ref{lemma:usefulksubmod}.
We constrain $z_j$ as follows.
\begin{equation}
\label{zj:unary}
z_j \geq f_j(0) + (2v_{i,d}-1)(f_j(d)-f_j(0)) \quad \forall d \in [k].
\end{equation}
%$(z_j \geq f_j(0) + (2v_{i,d}-1)(f_j(d)-f_j(0)))$ for each $d \in [k]$.
If $f_j$ is the soft version of $(v_p = \pi(v_q))$, for some
permutation $\pi$ on $[k]$, 
%Extend $f_j$ to the larger domain by
%$f_j(0,0)=0$ and $f_j(a,0)=f_j(0,a)=\half$ for $a \neq 0$,
%again as in Lemma~\ref{lemma:usefulksubmod}.
constrain $z_j$  as follows.
\begin{equation}
\label{zj:bijection}
z_j \geq |v_{p,\pi(d)} - v_{q,d}| \quad \forall d \in [k].
\end{equation}
Here, $z \geq |x-y|$ is shorthand for the two separate equations $z \geq
x-y$ and $z \geq y-x$. 
%by $(z_j \geq |v_{p,\pi(d)} - v_{q,d}|)$ for all $d \in [k]$, where the
%absolute value is implemented by two linear constraints.
If $f_j$ is the soft version of $(v_p=a \lor v_q=b)$
for some $a, b \in [k]$,
%As in Lemma~\ref{lemma:usefulksubmod},
%we extend $f_j$ to value $0$ by $f_j(a,0)=f_j(0,b)=f_j(0,0)=0$,
%with $f_j(d,0)=f_j(0,d)=\half$ for every other case.
constrain $z_j$ as follows.
\begin{equation}
\label{zj:lor}
z_j \geq 1-v_{p,a}-v_{q,b}.
\end{equation}
%by $(z_j \geq 1-v_{p,a}-v_{q,b})$.
Recall that $z_j \geq 0$ is additionally always in effect.
Finally, if $f_j$ is the soft wide equality $(v_{i_1}=\ldots=v_{i_r})$, for 
some $i_1, \ldots, i_r \in [n]$,
%Extend $f_j$ to value 0 as in Lemma~\ref{lemma:ksubmod}...
constrain $z_j$ as follows.
\begin{equation}
\label{zj:broadeq}
z_j \geq |v_{i_p,d} - v_{i_q,d}| \quad \forall d \in [k], p,q \in [r].
\end{equation}
Again, the absolute value is shorthand for a split into two equations. 
%by $(z_j \geq |v_{i_p,d} - v_{i_q,d}|)$ for all $d \in [k], p,q \in [r]$;
%the absolute value is implemented as before. 
This completes the description of the new LP. 
We will now show its half-integrality. 
The proof goes through a series of exchange arguments, but ultimately
the result comes down to showing that the new LP has an optimum 
which corresponds exactly to an integral optimum of the basic LP,
using the relaxation functions of Lemma~\ref{lemma:usefulksubmod}.

We need some terminology. 
Let $v_i \in V$ be a variable of the CSP, and let
$v_i^*:=(v_{i,1},\ldots,v_{i,k})$ denote the vector of corresponding
variables in the above LP. We say that $v_{i,d}$ is \emph{tight} in an
assignment if there exists some $d' \in [k], d \neq d'$ such that
$v_{i,d}+v_{i,d'}=1$, and that $v_i$ has a \emph{standard assignment} 
if $v_{i,d}$ is tight for every $d \in [k]$. 
Thus in a standard assignment, $v_i^*$ is
characterised by the \emph{mode} $\arg \max_{d \in [k]} v_{i,d}$ 
and its \emph{frequency} $\max_{d \in [k]} v_{i,d}$. 
An assignment $v_i=d$ in the CSP, for $d \neq 0$,
corresponds to a standard assignment with mode $d$ and frequency $1$,
while an assignment $v_i=0$ in the CSP corresponds to a standard
assignment with frequency $\half$. Let the \emph{half-integral} standard
assignments be those whose frequency is either $\half$ or $1$. 

We give the proof in two parts, first showing that there is an LP-optimum
where every variable vector $v_i^*$ takes a standard assignment, then
showing that in fact, this assignment can be taken to be half-integral. 
By further observing that in a half-integral assignment, each 
cost variable $z_j$ takes the value of the corresponding $k$-submodular
2-relaxation of Lemma~\ref{lemma:usefulksubmod}, we complete the proof. 

\begin{lemma}
\label{lemma:lp-standard-assignment}
Let $\phi^*$ be an optimum to the above LP, and let $X$ be the set of
variables $v_{i,d}$ which are not tight in $\phi^*$, and such that
$v_{i,d}<\half$. Let $\phi'(\varepsilon)$ equal $\phi^*+\varepsilon X$,
with variables $z_j$ readjusted accordingly. Then for some
$\varepsilon>0$, $\phi'(\varepsilon)$ is another 
optimal assignment to the LP. 
\end{lemma}
\begin{proof}%[Proof of Lemma~\ref{lemma:lp-standard-assignment}]
%Let $\xi$ be the largest value of $\varepsilon$ such that 
%$\phi'(\varepsilon)$ satisfies the framework constraints of the LP, 
%and such that $v_{i,d}\leq \half$ for every $v_{i,d}\in X$. 
By readjusting the variables $z_j$, we mean that every variable $z_j$ is
given the smallest possible feasible value, given the assignments to the
variables $v_{i,d}$ fixed by $\phi'(\varepsilon)$. 
Clearly, there is some $\varepsilon>0$ such that $\phi'(\varepsilon)$ 
is a feasible assignment; we will further verify that the readjustment of
the variables $z_j$ does not increase the total cost. 
This is done on a constraint-by-constraint basis.
%or every variable $z_j$, let $z_j$ refer to its value in $\phi$ and
%z_j'$ to its value in $\phi'( $. 

\begin{claim}
\label{claim:unary}
Let $f_j$ be a unary cost function on a variable $v_i$ in the CSP,
and $z_j$ constrained as in (\ref{zj:unary}). 
For a sufficiently small $\varepsilon>0$, the value of $z_j$ does not
increase. 
\end{claim}
\begin{proof}
Let $d_1$ and $d_2$ be the first and second minimising values of $f_j$, as above.
We assume for simplicity that $f_j(0)=0$ (even at the risk of having
$f_j(d_1)<0$), by adjusting every value of $f_j(\cdot)$ by $-f_j(0)$.
Observe that the value of $z_j$ changes by this by only a constant. 
We also readjust $z_j \geq 0$ to $z_j \geq -f_j(0)$; thus this is a simple
shift of the value of $z_j$.
We can simplify (\ref{zj:unary}) as follows:
\[
z_j \geq (2v_{i,d}-1)f_j(d) \quad \forall d \in [k].
\]
First assume that $f_j(d_1)=f_j(d_2)=f_j(0)=0$; thus $f_j(d)\geq 0$
for every $d$. In particular, for $d=d_1$ the equation reads $z_j \geq 0$.
To raise the value of $z_j$, some variable $v_{i,d}$ must have a value
greater than $\half$, but such a variable would not be changed.

Now, assume that we have $f(d_1)<0$, thus $f(d_1)+f(d_2)=0$. 
If $v_{i,d_1}<\half$ %\ynote{$v_{i,d_1}$?} 
then $z_j > 0$, but raising the value of $v_{i,d_1}$ does not
increase $z_j$; in this case, the only other possible tight value for $z_j$
would be some $d$ such that $v_{i,d}>\half$, but again, such a variable
would not be readjusted. 

Otherwise $v_{i,d_1}\geq \half$, but then $v_{i,d} \leq 1-v_{i,d_1} \leq \half$ 
for every $d \neq d_1$. Inserting $d=d_2$ into the equation we have
a right-hand-side of $(2v_{i,d_2}-1)f_j(d_2) \leq (1-2v_{i,d_1})f_j(d_2)
= (2v_{i,d_1}-1)f_j(d_1)$, matching the equation for $d=d_1$; 
for every other value of $d$, the equation has at least as high slope. 
Thus no non-tight value other than $d_1$ can define the value of $z_j$. 
\maybeqed{} \end{proof}

\begin{claim}
\label{claim:bijection}
Let $f_j$ be the soft version of the constraint $(v_p = \pi(v_q))$,
and $z_j$ constrained as in (\ref{zj:bijection}).
For a sufficiently small $\varepsilon>0$, the value of $z_j$ does not
increase. 
\end{claim}
\begin{proof}
Assume that $v_{q,b}$ is raised, immediately increasing the value of
$z_j$. Let $a=\pi(b)$. Then $v_{p,a}$ cannot be raised by $X$, hence
either $v_{p,a} \geq \half$ or $v_{p,a}$ is a tight value. But since
$v_{q,b}<\half$, in the former case the value of $z_j$ will not increase;
hence $v_{p,a} \leq v_{q,b}$ and $v_{p,a}+v_{p,a'}=1$ for some $a'\in [k]$.
Let $b'=\pi^{-1}(a')$. Then $v_{p,a'}-v_{q,b'} > (1-v_{p,a}) - (1-v_{q,b})
= v_{q,b}-v_{p,a}$, contradicting the claim that the equation
$v_{q,b}-v_{p,a}$ maximises $z_j$. 
\maybeqed{} \end{proof}

\begin{claim}
\label{claim:lor}
Let $f_j$ be the soft version of the constraint $(v_p=a \lor v_q=b)$,
and $z_j$ constrained as in (\ref{zj:lor}).
For a sufficiently small $\varepsilon>0$, the value of $z_j$ does not
increase. 
\end{claim}
\begin{proof}
The right-hand-side of (\ref{zj:lor}) has no positive coefficients for any $v_{i,d}$. 
\maybeqed{} \end{proof}

\begin{claim}
\label{claim:broadeq}
Let $f_j$ be the soft equality $(v_{i_1}=\ldots=v_{i_r})$, for
some $i_1, \ldots, i_r \in [n]$,
and let $z_j$ be constrained as in (\ref{zj:broadeq}).
For a sufficiently small $\varepsilon>0$, the value of $z_j$ does not increase. 
\end{claim}
\begin{proof}
Note that the value of $z_j$ equals the largest cost of a soft binary 
equality $(v_p=v_q)$ for $p, q \in \{i_1,\ldots,i_r\}$. 
By Claim~\ref{claim:bijection}, for a sufficiently small $\varepsilon>0$,
no such binary equality increases in cost, hence neither does $z_j$.
\maybeqed{} \end{proof}

Thus, for every constraint $f_j$ there is some value $\varepsilon>0$
such that $\phi'(\varepsilon)$ does not incur a larger cost for $f_j$ than
$\phi^*$. Since this is a finite number of bounds, taking the minimum
still yields some $\varepsilon>0$ and the proof finishes. 
\maybeqed{} \end{proof}

This implies that there is some LP-optimum $\phi^*$ such
that computing $X$ from $\phi^*$ yields an empty set.
(This follows by, e.g., considering that optimum
$\phi^*$ which maximises $\sum_{i,d} v_{i,d}$.)
In such an LP-optimum $\phi^*$, every variable $v_{i,d}$ with $v_{i,d}<\half$ is tight, 
and hence every variable $v_{i,d}$ is tight (by consider a corresponding variable $v_{i,d'}\leq \half$),
%\ynote{How do you derive this?}, is tight, 
i.e., $\phi^*$
is a standard assignment. We proceed to show that there is a half-integral optimum.

\begin{lemma}
\label{lemma:pinchspred}
Let $\phi^*$ be an optimum which is a standard assignment. 
Let $X^+=\{v_{i,d}: 1>\phi^*(v_{i,d})>\half\}$
and $X^-=\{v_{i,d}: 0<\phi^*(v_{i,d})<\half\}$.
For some sufficiently small $\varepsilon>0$,
we have that $\phi^*+\varepsilon(X^+-X^-)$
and $\phi^*-\varepsilon(X^+-X^-)$ are both optimal assignments. 
\end{lemma}
\begin{proof}%[Proof of Lemma~\ref{lemma:pinchspred}]
It is clear that both suggested assignments are feasible and standard
for some $\varepsilon>0$.
Let $\phi'(\xi)=\phi^*+\xi (X^+-X^-)$; we will verify that there is some 
$\varepsilon>0$ such that for every constraint $f_i$, the cost of $f_i$ is
a linear function in $\xi$ for $|\xi|\leq \varepsilon$.
Since $\phi^*$ is an optimal assignment, this must imply that all these
linear cost functions cancel and the cost is invariant under $\xi$. 
We again proceed by type of constraint.

\begin{claim}
\label{claim2:unary}
Let $f_j$ be a unary cost function on a variable $v_i$ in the CSP.
For a sufficiently small $\varepsilon>0$, the value of $z_j$ 
is locally linear in $\xi$. 
\end{claim}
\begin{proof}
Let $v_i$ be the involved variable, and let $d$ be the mode of $v_i$.
We assume that $v_i$ is not already half-integral (since then, $v_i$ would
be kept constant). 
Let $d_1, d_2$ be the two minimising values, as before. If $f(d)>f(d_2)$,
then the equation for value $d$ is the sole maximising equation for $z_j$,
which is thus locally linear. If $d=d_1$, then the maximising equations
are for values $d_1$ and any $d'$ such that $f_j(d')=f_j(d_2)$. 
If the former instantiation of equation~(\ref{zj:unary}) has slope $\alpha$, 
then all latter instantiations have
slope $-\alpha$, thus modification by $\xi$ is locally linear. 
Otherwise, $d$ and $d_1$ are the unique maximising equations, and again
the slopes are each others' opposites, making $\xi$ locally linear. 
This finishes the claim.
\maybeqed{} \end{proof}

\begin{claim}
\label{claim2:bijection}
Let $f_j$ be the soft version of the constraint $(v_p = \pi(v_q))$. 
For a sufficiently small $\varepsilon>0$, the value of $z_j$ 
is locally linear in $\xi$. 
\end{claim}
\begin{proof}
Let $a$ be the mode of $v_p$ and $b$ be the mode of $v_q$. 
Observe that the cost of $z_j$ equals $|v_{p,a}-v_{q,b}|$ if $a=\pi(b)$, 
otherwise $v_{q,b}-v_{p,\pi(b)} =
(1-v_{q,\pi^{-1}(a)})-(1-v_{p,a})=v_{p,a}-v_{q,\pi^{-1}(a)} = z_j$,
and the latter holds for any standard assignments to $v_p$ and $v_q$. 
If one variable, say $v_q$, is already half-integral, then this yields a
linear function (in particular as the absolute value in the first case is
non-zero, given that $v_q$ is half-integral but $v_p$ not). 
If both variables are fractional, the first case applies, and
$v_{p,a}=v_{q,b}$, then observe that $v_{p,a}$ and $v_{q,b}$ are modified
identically by $\xi$. Finally, in any other case $z_j$ is determined by a
locally linear function of the involved variables $v_{p,d}, v_{q,d}$. 
\maybeqed{} \end{proof}

\begin{claim}
\label{claim2:lor}
Let $f_j$ be the soft version of the constraint $(v_p=a \lor v_q=b)$. 
For a sufficiently small $\varepsilon>0$, the value of $z_j$ is locally
linear in $\xi$. 
\end{claim}
\begin{proof}
Let $z=1-v_{p,a}-v_{q,b}$. If $z>0$, then $z_j=z$ and $z_j$ is determined
solely by this equation. If $z<0$, then $z_j=0$ up to some local
adjustment $\xi$. Finally, if $z=0$, either $v_p$ and $v_q$ are both
half-integral, and $z_j$ is constant in $\xi$, or $v_p$ and $v_q$ are
adjusted by $\xi$ in opposite directions, again leaving $z_j$ constant. 
\maybeqed{} \end{proof}

\begin{claim}
\label{claim2:broadeq}
Let $f_j$ be the soft equality $(v_{i_1}=\ldots=v_{i_r})$, for
some $i_1, \ldots, i_r \in [n]$. For a sufficiently small $\varepsilon>0$,
the value of $z_j$ is locally linear in $\xi$. 
\end{claim}
\begin{proof}
W.l.o.g., let us use $i_t=t$ for each $t \in [r]$. Observe that for every
variable $v_p, p \in [r]$, with mode $d$, the cost of the pair $(v_p,v_q)$
equals $|v_{p,d}-v_{q,d}|$ for every other variable $v_q, q \in [r]$.

Let $v_1$ be a variable among the set which maximises the frequency (i.e.,
if any variable is integral, then $v_1$ is integral). Let $a$ be the mode
of $v_1$. Let $v_r$ be a variable which minimises $v_{p,a}$, thus
$(v_1,v_r)$ maximises the cost of $z_j$. 

If $v_{r,a} \geq \half$, then observe that no variable $v_p$ for $p\in[r]$
has a mode other than $a$, hence the tight pairs are exactly pairs
$(v_p,v_q)$ where $v_{p,a}=v_{1,a}$ and $v_{q,a}=v_{r,a}$. If $v_1$ is
integral and $v_r$ half-integral, then this cost is unaffected by $\xi$;
if $v_1$ is integral but $v_r$ is not half-integral or vice versa, 
then the cost is a linear function of $\xi$; and if neither case occurs,
then for every pair of LP variables $(v_{p,a},v_{q,a})$, the pair are
adjusted equally by $\xi$ and $z_j$ is constant. 
This finishes the case $v_{r,a}\geq \half$.

Thus assume that $v_{r,a}<\half$, and let $b$ be the mode of $v_r$. 
If $v_{r,b}=v_{1,a}$, then edges which maximise $z_j$ go only between
variables of this frequency; either this frequency is $1$, in which case
we have contradictory integral assignments and $z_j=1$ independent of
$\xi$, or $\xi$ modifies all these maximal frequencies identically, thus
the situation is preserved by the modification and the cost is modified
linearly in $\xi$. 

Otherwise, let $U$ be all variables $v_i$ such that $v_{i,a}=v_{1,a}$, 
and let $W$ be all variables $v_i$ such that $v_{i,a}=v_{r,a}$. 
The pairs $(v_p,v_q)$ which maximise $z_j$ are exactly those where $v_p\in
U$ and $v_q \in W$, furthermore, the cost of such an edge is exactly
$v_{p,a}-v_{q,a}$ (by the initial observation). Furthermore, this
situation is preserved by some local variation of $\xi$; our conditions
are $v_{1,a} > \half > v_{r,a}$ and $v_{1,a} > v_{r,b}$, both of which are
stable for some range of $\xi$. Finally we observe that all costs
$v_{p,a}-v_{q,a}$ in fact equal $v_{1,a}-v_{r,a}$ also after modification
by $\xi$, hence $z_j$ is locally linear. 
\maybeqed{} \end{proof} 

Since every constraint $f_j$ is found to have locally linear cost while  
$|\xi| \leq \varepsilon$ for some $\varepsilon>0$, and since there is a
finite number of constraints, there is some $\varepsilon>0$ such that
$|\xi| \leq \varepsilon$ implies that every constraint $f_j$ varies
linearly with $\xi$. By optimality of $\phi^*$, the total cost must thus
be locally constant. 
\maybeqed{} \end{proof}

We can now finish our result.

\begin{theorem}
\label{thm:newlp-halfint}
The above LP has a half-integral optimum, which can be found in polynomial
time, and which corresponds directly to an optimal assignment for the
original CSP. 
\end{theorem}
\begin{proof}
Let $x^*$ be the optimal value of the above LP, and let $\phi^*$ be an
assignment which achieves this cost, and subject to this maximises
$\sum_{i=1}^n \sum_{d=1}^k v_{i,d}$. Then $\phi^*$ must be a standard
assignment by Lemma~\ref{lemma:lp-standard-assignment}.
Furthermore, we must have $X^+=X^-=\emptyset$ as computed in
Lemma~\ref{lemma:pinchspred}: otherwise, by Lemma~\ref{lemma:pinchspred}
some ``local adjustment'' $\xi$ is possible, but for $k>2$ such an
adjustment would strictly increase the optimum of the secondary LP.  
Thus $\phi^*$ is a standard assignment with $X^+=X^-=\emptyset$, i.e.,
half-integral. 

For the last part, simply verify for each of the four constraint types
that the cost $z_j$ when evaluated at a half-integral point equals exactly
that of the $k$-submodular relaxations given in
Lemma~\ref{lemma:usefulksubmod}. 
\maybeqed{} \end{proof}

\subsection{The parameterized complexity of Unique Label Cover} 
\label{section:ulc}

We now focus specifically on consequences for the problem 
\textsc{Unique Label Cover}. This is the defining problem of the
Unique Games Conjecture~\cite{Khot02}, which is of central importance to
the theory of approximation. In our terms, \textsc{Unique Label Cover} 
corresponds to the problem VCSP$(\cF)$ where $\cF$ contains the soft versions of
all constraints $(x=\pi(y))$ for bijections $\pi$ on a domain $D=[k]$ for some
$k$. 
In the below, we will consider both edge- and vertex-deletion versions of
the problem; we will let $\Sigma$ denote the label set of an instance
(corresponding to the domain $D$), and $p$ the minimum instance cost
(i.e., the minimum number of edges resp.\ vertices one needs to delete to
get a satisfiable remaining instance). Observe that there is a simple
reduction from the edge-deletion version to the vertex-deletion version. 
The problem was previously considered from an FPT perspective by Chitnis~et~al.~\cite{ChitnisCHPP12}, who provided an FPT algorithm in the two
parameters $|\Sigma|,p$, with a running time of $\Oh^*(|\Sigma|^{O(p^2 \log
  p)})$, using highly advanced algorithmic methods. We observe that we can
improve the running time.

\begin{corollary} 
\label{cor:ulc-alg}
\textsc{Unique Label Cover} is FPT, both in edge- and vertex-deletion
variants, with a running time of $\Oh^*(|\Sigma|^{2p})$, where $\Sigma$ is
the label set of the instance and $p$ is its cost (i.e., the minimum
number of non-satisfied edges resp.\ vertices). 
\end{corollary}
\begin{proof}
For the edge deletion case, the result follows directly from the basic LP
relaxation (e.g., invoking Corollary~\ref{cor:ksubmod:direct} using
constraint set $\cF$ as above and relaxations given by
Lemma~\ref{lemma:usefulksubmod}). 

For vertex deletion, we follow the outline sketched in Section~\ref{section:vdel}.
For every edge-constraint in the input, we
create $2p+1$ copies of the corresponding soft constraint, to make it too
costly to break. For every vertex $v \in V$, we split $v$ into $t:=d(v)$
copies $v(1), \ldots, v(t)$, and place one such copy in every edge $uv$
involving the vertex $v$ (and hence in all $2p+1$ valued constraints
stemming from the edge). Finally, we introduce a soft equality constraint
$(v(1)=\ldots=v(t))$, which can be broken at cost 1 with a net effect
equivalent to that of deleting $v$.

To solve this problem, we can then invoke the generic result of
Lemma~\ref{lm:relaxfpt}, using the $k$-submodular relaxations of
Lemma~\ref{lemma:usefulksubmod} and the LP-formulation given in
Section~\ref{section:betterlp} (due to Theorem~\ref{thm:newlp-halfint}). 
\maybeqed{} \end{proof}

Chitnis~et~al.~\cite{ChitnisCHPP12} showed that the problem is W[1]-hard,
even in the edge-deletion version, when parameterized by $p$ alone %\ynote{should this be $p$?}
(when $|\Sigma|$ occurs in the input) by a reduction from \textsc{$k$-Clique}.
This implies a conditional lower bound on the running time via the ETH-hardness of 
\textsc{$k$-Clique} (see~\cite{ChenCFHJKX05,ChenHKX06}); however, 
despite the above improvement, the upper and lower bounds still do not meet. 
We leave it as an open question whether a running time like $O^*(c^p|\Sigma|^{o(p)})$ 
would contradict ETH.

Finally, we observe that the improved branching used in Section~\ref{section:gfvs}
for \textsc{Group Feedback Vertex Set} partially applies here, implying a
running time bound of $O(4^p|\Sigma|^c)$, where $c$ is the number of
connected components after OPT has been removed. (In particular, for
the edge-deletion version we may slightly refine this to $2^{2p-c}|\Sigma|^c)$, 
and observe $c \leq p+1$, assuming that $G$ is connected.) 

\section{Group Feedback Vertex Set}
\label{section:gfvs}

We now consider the application of the above techniques to the problem of
\textsc{Group Feedback Vertex Set}. We first review a few notions
(essentially following Guillemot~\cite{Guillemot11} and Cygan et
al.~\cite{CyganPP12}). 
Let $\Gamma$ be a finite group with identity element $1_{\Gamma}$. 
A \emph{$\Gamma$-labelled graph} is a graph $G=(V,E)$ with a labelling
$\lambda: E \rightarrow \Gamma$ such that
$\lambda(u,v)\lambda(v,u)=1_{\Gamma}$ for every edge $uv \in E$.
A \emph{consistent labelling} for a $\Gamma$-labelled graph $G$ is a
labelling $\phi: V \rightarrow \Gamma$ such that for every $uv \in E$,
$\phi(u)\lambda(u,v)=\phi(v)$. 
We now define the problem.

\parameterizedproblem{Group Feedback Vertex Set}{A group $\Gamma$, a 
$\Gamma$-labelled graph $G=(V,E)$ with labelling $\lambda$, and an integer $k$.}%
{$k$}{Is there a set $X \subseteq V$ with $|X|\leq k$ such that $G
  \setminus X$ has a consistent labelling?}

For a path $P=v_1 \ldots v_r$, we let
$\lambda(P)=\lambda(v_1,v_2) \cdot \ldots \cdot \lambda(v_{r-1},v_r)$;
similarly, for a cycle $C=v_1v_2 \ldots v_rv_1$, we let
$\lambda(C)=\lambda(v_1,v_2) \cdot \ldots \cdot \lambda(v_r,v_1)$. 
We say that $C$ is \emph{non-null} if $\lambda(C) \neq 1_{\Gamma}$. 
An important aspect of the problem is the following ``dual'' view on
consistency. 
%(In particular, while the choice of starting vertex of $C$
%may matter for the resulting group element when the group is not Abelian,
%the choice does \emph{not} matter for the conclusion of whether the cycle
%is null or non-null.)

\begin{lemma}[\cite{Guillemot11}]
\label{lemma:guillemot}
A $\Gamma$-labelled graph $G$ has a consistent labelling if and only if it
contains no non-null cycles.
\end{lemma}

Since the consistency condition simply needs to verify the bijections on
the edges, the \textsc{Group Feedback Vertex Set} problem is a special
case of \textsc{Unique Label Cover}, and is thus covered by the result of
Section~\ref{section:ulc}. However, it turns out we can do much better. 
The following will be the main conclusion of the current section. 

\begin{theorem}
\label{theorem:gfvs}
The \textsc{Group Feedback Vertex Set} problem can be solved in time
$O^*(4^k)$, even when the group $\Gamma$ is given via oracle access only. 
\end{theorem}

Previous work by Guillemot~\cite{Guillemot11} and by Cygan et
al.~\cite{CyganPP12} established that the problem is FPT, however, the
best achieved running time was $O^*(2^{O(k \log k)})$~\cite{CyganPP12}. 
We follow Cygan~et~al.~\cite{CyganPP12} in the definitions of the oracle
access model: we assume that we have access to an oracle which can
multiply two elements, invert an element, produce the identity element
$1_{\Gamma}$, and verify whether two elements are equal.

\subsection{An improved branching algorithm}
\label{section:improved}

We begin by describing the improved branching process that lies behind
Theorem~\ref{theorem:gfvs}. We assume that $\Gamma$ is given via
%by explicit description; the full case of 
oracle access, e.g., we are dealing with VCSP$(\cF)$ for a humongous domain $\Gamma$. %\ynote{Why is this related to fixing assignments?}).
% will be handled later. 
Let \textsc{GFVS with Assignments} for group $\Gamma$
denote 
\textsc{Group Feedback Vertex Set} enhanced with a requirement that
certain variables take certain values in the optimum.
%(e.g., 
%VCSP$(\cF)$ where $\cF$ contains the soft equality
%constraint of every arity, and all soft bijection constraints
%$(u=v\cdot\gamma)$ and assignments $(v=\gamma)$ for each
%$\gamma \in \Gamma$.  
Furthermore, let \textsc{Half-integral GFVS with
  Assignments} refer to the $k$-submodular 2-relaxation of this problem, 
as given by Lemma~\ref{lemma:usefulksubmod}.
% (this can easily be defined in the oracle model
In the following, we assume that each invocation of \textsc{Half-integral
  GFVS with Assignments} returns an optimal solution 
(rather than just a cost). 

\begin{lemma}
\label{lemma:gfvs:invoke}
\textsc{Group Feedback Vertex Set} can be solved via $\Oh^*(4^k)$
invocations of \textsc{Half-integral GFVS with Assignments}.
\end{lemma}
\begin{proof}
The improvement is centred around the following observation.

\begin{claim}
Let $(G,\Gamma,\lambda,k)$ be an instance of \textsc{Group Feedback Vertex
  Set} (without assignments). Let $v \in V$ be an arbitrary vertex. Then
either $v$ is deleted by every optimal solution, or there is an optimal
solution with a consistent labelling $\phi$ where $\phi(v)=1_{\Gamma}$. 
\end{claim}
\begin{proof}
Let $X \subseteq V$ be an optimal solution with $v \notin X$, and let
$\phi$ be the corresponding consistent labelling. Then for any $\gamma \in
\Gamma$, $\phi'(u)=\phi(u) \cdot \gamma$ defines another consistent
labelling of the graph. In particular, we can choose
$\gamma=\phi(v)^{-1}$. 
\maybeqed{} \end{proof}

Thus, we initialise our algorithm by picking an arbitrary $v \in V$, and
branch on deleting $v$ (at cost 1) or fixing an assignment
$(v=1_{\Gamma})$ (at cost at least \half, assuming the input is not already
consistent). We will grow a connected component of integrally assigned
vertices, in each iteration selecting a new relaxed vertex $v$
neighbouring this component, and branch on whether $v$ is deleted or not.
Whenever we ``run out'' of such candidate vertices $v$, we simply restart
the process with a new arbitrary assignment.

Concretely, we do the following. As before, we split every vertex into
different variables $v(i)$ for all its edge occurrences, then replicate
each edge constraint $2k+1$ times to prevent edges from being broken. We
maintain a set $A$ of enforced assignments $(v(i)=d)$ and a set $X$ of
explicit deletions, both initially empty. We let $k_0 \gets k$ be our
initial budget bound. Our branching algorithm then proceeds as follows:
Let $\phi^*$ be an optimal solution for the \textsc{Half-integral GFVS
  with Assignments} instance corresponding to $G$, $X$ and $A$ (where $X$
is implemented by simply omitting the corresponding soft equality
constraints from the instance construction), and let $x^*$ be the cost of
$\phi^*$. Compute $k=k_0-|X|-x^*$; if $k_0<0$, reject. Add to $A$ any
integral assignments of $\phi^*$ not already contained in it, and add to
$X$ any variables $v$ such that $A$ contains $(v(i)=d)$ and $(v(j)=d')$
for some integral values $d \neq d'$. If there is a \emph{half-deleted}
vertex $v$ (i.e., a vertex such that the cost of its soft equality
constraint is $\half$), let $d$ be the non-zero value assigned to some
occurrence of $v$. Compute two new instances, one where assignments
$(v(i)=d)$ are added to $A$ for all occurrences of $v$, and one where $v$
is added to $X$. If either of these instances does not lead to a decreased
budget, then we claim that the corresponding solution must contain at
least one new integral assignment $u(i)=d$, $u \in V$. In the former case,
this is clear; in the latter case, observe that replacing $v$ from $X$
into the instance as a soft equality constraint yields a valid relaxed
solution, thus it must be that $v$ uses two distinct integral assignments
in the new relaxed optimum (note that $v(i)=d$ for some $i$ due to
assignments in $A$). Finally, if both new instances lead to a decrease in
$k$, branch accordingly in both directions.

The remaining case is that every vertex $v$ is either fully deleted or not
deleted at all in the current optimal relaxed assignment. But then, all
assigned vertices form connected components, whose every neighbour in the
original graph $G$ is contained in $X$. In other words, the remaining
graph $G \setminus X$ contains a connected component of entirely relaxed
vertices; we may then pick an arbitrary occurrence $v(i)$ of an unassigned
vertex $v$, and add $(v(i)=1)$ to $A$ (leading ultimately to a solution where
$v$ is either fully assigned or fully deleted).

Throughout, the correctness of our operations rely upon the persistence of
the relaxation \textsc{Half-integral GFVS with Assignments}. The branching
tree has a branching factor of $2$, and a depth of at most $2k$, and in
every node we make a polynomial number of calls to \textsc{Half-integral
  GFVS with Assignments}. 
\maybeqed{} \end{proof}

By the above, we get an $O^*(4^k)$-time algorithm for \textsc{Group
  Feedback Vertex Set} when the group $\Gamma$ is given explicitly, e.g., via invocation of Theorem~\ref{thm:newlp-halfint} of
Section~\ref{section:betterlp} to solve the \textsc{Half-integral GFVS with
  Assignments} subproblems.  
The case of oracle-access only to $\Gamma$ is handled next.

\subsection{Oracle-access groups}

Unlike in the last section, when $\Gamma$ is given only via oracle access
it could be that $\Gamma$ contains an exponential number of elements
(indeed, the simplest reduction from \textsc{Feedback Vertex Set} uses the
group $\Gamma=Z_2^m$). To handle this, we redesign the LP to not
keep track of vertices' explicit assignment, but only whether each vertex
has been deleted or not. We introduce one variable $z_i$ for each vertex
$v_i \in V$, and an exponential number of constraints (solved via a
separation oracle) as follows. By a simple reduction, assume that a
unique assignment $(t=1_{\Gamma})$ is required. A \emph{double path} ending in $v\in
V$ is a pair $(P_a, P_b)$ of paths from $t$ to $v$, such that 
$\lambda(P_a) \neq \lambda(P_b)$. Let $z(P)$ denote the sum of $z_i$ for
\emph{internal} vertices of a path $P$. Then the \emph{length} of
$(P_a,P_b)$ is defined as $z(P_a,P_b):=z(P_a)+z(P_b)+z(v)$ (note that
internal vertices common to both paths are counted twice). 
The length of a cycle $C$ is defined as $z(C)=\sum_{v_i\in C} z_i$. 
For simplicity, for a vertex $v=v_i$, we write $z(v)$ for $z_i$ 
(to avoid having to explicitly state all vertex indices). 
Our constraints will state that the length of every double path is at least $1$. 
Call the resulting constraints a \emph{double path system}. 
A set of weights $z_i$ under which every double path has length at least $1$
is said to be \emph{double-path-hitting}.
We will show that double path systems can be used to solve the 
\textsc{Half-Integral GFVS with Assignments} problem (half-integral GFVS, for short), 
even for groups with oracle access, which then combined with Lemma~\ref{lemma:gfvs:invoke} 
yields an FPT algorithm for GFVS. 

%For reasons of space, we omit all proofs of this section.

We now proceed with the proofs. We first show that vertex-deletion information 
is sufficient, then we show that the double path system actually provides
this information.

\begin{lemma}
\label{lemma:gfvs-del-is-enough}
Assume a solver for \textsc{Half-integral GFVS with Assignments} which
reveals the costs of the soft equality constraints of the instance, but no
more information. From this we can construct an optimal assignment.
\end{lemma}
\begin{proof}
Clearly, we must satisfy all assignments from $A$. Furthermore, we may let
these assignments propagate through edge labels until we reach a vertex in
the support of the half-integral solution (i.e., partially or fully deleted). %\ynote{Of what is the support?}. 
In this case, we fix the
assignment to the corresponding occurrence $v(i)$ of this vertex $v$, but
do not propagate further through $v$. If this leads to a contradictory
assignment (other than for fully deleted vertices), then the deletion
values did not encode a feasible assignment. Otherwise, after this process
terminates, we may safely assign every other variable the value $0$. 
\maybeqed{} \end{proof}

\subsubsection{Equivalence of the formulations}
To show that double path systems solve half-integral GFVS, 
we show that they are (in an appropriate sense) equivalent to the improved LP
formulations of Section~\ref{section:betterlp}; the existence of 
a half-integral optimum then follows from Theorem~\ref{thm:newlp-halfint}. 
Refer to the LP of Section~\ref{section:betterlp} as the \emph{reference LP}.

We first show that every half-integral optimum of the reference LP 
satisfies all constraints of the double path system.

\begin{lemma}
\label{lemma:double-path-necessary}
Let $\phi^*$ be a half-integral optimum to the reference LP corresponding
to an instance of \textsc{Half-integral GFVS with Assignments}, and let
$z_i$ be the weight in $\phi^*$ of the soft equality constraint for $v_i$,
for each $i \in [n]$. Then these values $z_i$ are double-path-hitting. 
Furthermore, every other soft constraint in the original LP has cost zero
under $\phi^*$. 
\end{lemma}
\begin{proof}
We begin by the last point: By the construction of the VCSP, any optimal
solution will satisfy each assignment and each edge constraint
$(v(i)=\pi(u(j))$ at cost zero; thus the only constraints not completely
satisfied are the soft equality constraints. 

Now let $V_1=\{v_i \in V: z_i=1\}$ and $V_{\half}=\{v_i \in V: z_i=\half\}$. 
Let $H$ be the connected component of $G$ induced by the vertices
reachable from $t$ in $G \setminus (V_1 \cup V_{\half})$. 
Then $H$ has a consistent labelling (as all constraints within $H$ are
satisfied). Thus, every double path must intersect $V_1$ or $V_{\half}$.
If a double path intersects $V_1$, or intersects $V_{\half}$ in two places
or in a vertex with multiplicity two in the double path, then certainly the double path has
length at least $1$. Thus let $(P_a,P_b)$ be a double path, ending at $v$,
which intersects exactly one vertex $u \in V_{\half}$ (we may have $u=v$). 

If $u=v$, let $v(i)$ and $v(j)$ be the occurrences of $v$ at which the
paths $P_a$ and $P_b$ end. Since these paths (excluding the endpoint) 
are contained in $H$, the penultimate vertex of each path must be integral. 
But then $v(i)$ and $v(j)$ are both integral, and 
%
%Let $u_a$ resp.\ $u_b$ be the last vertices of
%$P_a$ resp.\ $P_b$ before $v$. Since the paths, excluding $v$, lie entirely
%in $C$, both $u_a$ and $u_b$ have integral assignments in $\phi^*$. 
%Since each edge constraint is satisfied, we find that $v(i)$ and $v(j)$
%have integral assignments in $\phi^*$, which 
by the inconsistency of the two paths these must be different. 
This contradicts the claim that $v \in V_{\half}$. 

Otherwise, assume w.l.o.g.\ that $z(P_a)=z(v)=0$, and $z(P_b)=\half$. 
Let $u(i)$ and $u(j)$ be the first and second occurrence of $u$ in $P_b$ 
(e.g., the occurrences of $u$ on the edge which enters resp.\ leaves $u$).
Since all vertices of the double path except $u$ are contained in $H$, we
have integral assignments to all variables, including $u(i)$ and $u(j)$,
and for every vertex $v' \neq u$ they are at cost zero. 
Thus, since the double path is inconsistent, $u(i)$ and $u(j)$ must have
distinct integral assignments, again contradicting that $u \in V_{\half}$. 
\maybeqed{} \end{proof}

We now show the reverse direction.

\begin{lemma}
\label{lemma:double-path-sufficient}
Let $z_i$ be an optimal assignment to the double path system corresponding
to an instance of \textsc{Half-integral GFVS with Assignments}. Then there
is a feasible assignment $\phi$ to the reference LP for the same instance,
where the cost of the soft equality for vertex $v_i$ is $z_i$, and where
all other soft constraints have cost zero. 
\end{lemma}
\begin{proof}
We will construct a feasible assignment $\phi$ to the reference LP,
where every vertex (or rather, every occurrence $v(i)$ of every vertex)
takes a standard assignment. To define this assignment, let $v(i)$ be an
occurrence of a vertex $v$ on an edge $uv$; temporarily treat $v(i)$ as a
vertex subdividing the edge $uv$, with $z(v(i))=0$, with an
identity-labelled edge connecting it to $v$. Let $P$ be a shortest path
from $t$ to $v(i)$ (as measured by $z(P)$), and let $\gamma \in \Gamma$ 
be its resulting label. 
If $z(P)\geq \half$, let $v(i)=0$; otherwise, let $v(i)$ take the fractional
assignment with mode $\gamma$ and frequency $1-z(P)$. Repeat this for
every occurrence $v(i)$ of every vertex $v$ of the graph.
We claim that this creates a feasible assignment to the reference LP,
where all constraints except soft equalities are satisfied, and the cost
of the soft equality corresponding to a vertex $v_j$ is at most $z_j$. 

For feasibility, we first need to verify that edge constraints are
satisfied at cost zero. Let $uv$ be an edge with corresponding vertex occurrences
$u(i), v(j)$. Observe that the length of the shortest paths to $u(i)$ and
to $v(j)$ are equal, as each path to the one is a valid path to the other;
%we assume that this distance is minimal, subject to $uv$ being an edge of non-zero cost. 
thus $u(i)$ and $v(j)$ have identical frequencies, 
and the question is if they have compatible modes. 
Let $P_u$ be the shortest path that led to the labelling of $u(i)$, 
and similarly let $P_v$ be the path to $v(j)$. 
Note that both paths have length less than $\half$. 
First, assume that $P_v$ passes through $u$ but not through $v$. 
Then the last edge of $P_v$ must be $uv$, and removing this edge
leaves two incompatible paths to $u$; furthermore, since $z(u)$ was
included in the cost of $P_v$, we have a double path of length less than $1$,
which contradicts $z_i$ being feasible. 
Otherwise, $P_u$ passes through $u$ and $P_v$ passes through $v$, 
thus the costs $z(u)$ resp.\ $z(v)$ are included in these.
Extending $P_u$ by the edge $uv$ now creates a double path ending in $v$,
of length less than one, again contradicting feasibility. 

Next, assume that $v_i$ is a vertex such that the cost of the soft
equality for vertex $v_i$ under $\phi$ (call this $c_i$) is more than $z_i$. 
Let $v_i(p), v_i(q)$ be two occurrences of $v_i$ maximising this cost,
and let $P_a$ resp.\ $P_b$ be corresponding shortest paths. 
If $v_i(p)$ and $v_i(q)$ have identical modes (or if at least one of them
takes value $0$), assume that the former has higher frequency. But then
$z(P_a) > z(P_b) + z_i$, 
%d(v_i(p)) > d(v_i(q))+z_i$ \ynote{The definition of $d(\cdot)$ should be given. Maybe better to denote by $\ell(\cdot)$ for consistency with the next section.}, 
which is a contradiction since the latter is the length of a possible path.  

Otherwise $v_i(p)$ and $v_i(q)$ have distinct modes. 
%; let $P_a$ and $P_b$ be the corresponding shortest paths. 
Then $(P_a,P_b)$ is a double path
ending in $v$. Now the cost $c_i>z_i$ equals $(1-z(P_a)) - (1-(1-z(P_b))) 
= 1-z(P_a)-z(P_b)$, i.e., the length of the double path is less than one,
again contradicting that $z_i$ are double-path-hitting. 
This finishes the proof. 
\maybeqed{} \end{proof}

We can conclude the following.

\begin{lemma}
\label{lemma:double-path-halfint}
The double path system has a half-integral optimum, and each such optimum
can be converted into an optimal solution for \textsc{Half-Integral GFVS
  with Assignments}. 
\end{lemma}
\begin{proof}
By Lemma~\ref{lemma:double-path-sufficient}, the cost of a set of
double-path-hitting weights is at least the cost of the reference LP; by 
Lemma~\ref{lemma:double-path-necessary}, the costs are in fact identical,
there is a half-integral optimum for the double path system, and every
such optimum can be interpreted as deletion values for an optimum for the
original LP. By Lemma~\ref{lemma:gfvs-del-is-enough}, we can reconstruct
an optimal full assignment for the VCSP from this information. 
\maybeqed{} \end{proof}

\subsubsection{Separation oracle and wrap-up}

It only remains to show that we can solve the double path system, i.e.,
that we can produce a polynomial-time separation oracle. This is not
difficult. Let us first show a structural result. (Recall that our
notion of path length $z(P)$ does not take into account the weight of
the end vertex of $P$.)

\begin{lemma}
\label{lemma:double-path-lollipop}
A set of weights $z_i$ is infeasible (i.e., fails to be
double-path-hitting) if and only if there is some non-null simple cycle
$C$, passing through a vertex $u$, such that $z(C)+2z(P_u)+z(u)<1$,
where $P_u$ is a shortest path to $u$. 
\end{lemma}
\begin{proof}
For a vertex $v$, let $\ell(v)=z(P_v)$ denote the length of a shortest path $P_v$ to $v$.
On the one hand, let $(P_a,P_b)$ be a double path of length less than $1$,
ending on $v$. If the paths are disjoint, then they form a non-null simple
cycle (passing through $t$, and we have $\ell(t)=0$). Otherwise, let $H$
be the graph consisting of the edges traversed by $P_a$ and $P_b$, with
edges used by both paths given multiplicity two. Observe that $H$ does not
admit a consistent labelling, thus by Lemma~\ref{lemma:guillemot},
$H$ contains a non-null simple cycle $C$. Further, $H$ is an even (Eulerian) graph
with maximum degree four, and the contribution of a vertex $u$ to the length
of the double path is $\frac 1 2 d_H(u) z(u)$ where $d_H(u)$ is the degree of $u$ in $H$. %\ynote{What's $d_H(u)$?}. 
Now, let $u_a$ resp.\ $u_b$ be 
the first vertices of $C$ reached by $P_a$ resp.\ $P_b$ (both exist, since
neither of $P_a$ or $P_b$ can contain all of $C$), and let $P_a'$ resp.\ $P_b'$
be the corresponding path prefixes. Observe that $u_a$ and $u_b$ both have 
multiplicity two in the double path (though we may have $u_a=u_b$).
Assume w.l.o.g.\ that $z(P_a')+z(u_a) \leq z(P_b')+z(u_b)$. 
The double path has length at least 
$z(P_a')+z(P_b')+z(u_b)+z(C) \geq 2z(P_a')+z(u_a)+z(C) \geq 2\ell(u_a)+z(u_a)+z(C)$;
hence $z(C)+2\ell(u_a)+z(u_a)<1$. 

On the other hand, let $C$ be a non-null simple cycle, and let $u$ be the
vertex of $C$ closest to $t$. Let $v$ be a vertex on $C$ other than $u$. 
Create one path $P_a$ going from $t$ to $u$ and further on to $v$ taking
one way around the cycle, and a path $P_b$ taking the same way from $t$ to
$u$ then further on to $v$ taking the other way around the cycle.
Then $(P_a,P_b)$ forms a double path of length exactly $2\ell(u)+z(u)+z(C) < 1$.
\maybeqed{} \end{proof}

Observe that it follows from the proof that there always exists a \emph{shortest} double path
$(P_a,P_b)$ such that $P_a+P_b=2P_u+C$ for some vertex $u$ and cycle $C$. 

\begin{lemma}
\label{lemma:double-path-oracle}
Double path systems have polynomial-time separation oracles. 
\end{lemma}
\begin{proof}
Let us assume that all shortest paths have distinct lengths;
this can be achieved by replacing each weight $z_i$ by the pair $(z_i,2^i)$
and handling weights lexicographically (e.g., treating $(z,b)$ as $z+b\varepsilon$
where $\varepsilon$ is infinitesimal). (The uniqueness now follows since shortest paths are induced.)
By this, we find that \emph{every} 
shortest double path $(P_a,P_b)$ contains one path, say $P_a$, which is 
the unique shortest path to the endpoint (as otherwise one of $P_a$ and $P_b$
could be replaced by the shortest path). 
By Lemma~\ref{lemma:double-path-lollipop}, we may also assume that $P_a+P_b$ forms a graph like $2P_u+C$ 
for some non-null cycle $C$. Pushing this further, we can conclude 
that for every vertex $v$ on $C$, the graph $P_a+P_b$ contains the shortest
path to $v$: For $u$, this is true by choice; for any other vertex $v$,
we may re-orient $P_a+P_b$ to end at $v$, and perform the above replacement. 
Thus, label every $v \in C-u$ by ``left'' or ``right'' according to whether 
the (unique) shortest path to $v$ goes clockwise or counterclockwise through $C$
after passing $u$ (give $u$ both labels). Let $vv'$ be an edge in $C$ whose 
endpoints have distinct labels (this exists, though one endpoint may be $u$).
By orienting $P_u+C$ to a double path ending in $v \neq u$, we get a (shortest) 
double path $(P_a,P_b)$ ending at $v$, where $P_a$ is the shortest path to $v$, 
and $P_b$ is the shortest path to $v'$. Thus finding a shortest double path
has been reduced to finding two vertices $v$ and $v'$, such that 
their total distance from $t$ (and their own weights) sum up to less than $(1,0)$,
and such that the resulting labels of the shortest paths are incompatible 
for the edge $vv'$. This can be done simply by computing shortest paths. 
\maybeqed{} \end{proof}

We may finally wrap up. 

%\vspace{1em}

\begin{proof}[Proof of Theorem~\ref{theorem:gfvs}.]
%\paragraph{} \noindent \emph{Proof of Theorem~\ref{theorem:gfvs}.}
By Lemma~\ref{lemma:gfvs:invoke}, it suffices to be able to produce an
optimal solution to \textsc{Half-integral GFVS with Assignments} in the
oracle access group model; by Lemma~\ref{lemma:double-path-halfint}, it
suffices to be able to produce a half-integral optimum to a double path
system. By Lemma~\ref{lemma:double-path-oracle}, double path systems can
be optimised in polynomial time. The only remaining detail is how to
convert an arbitrary optimum to a double path system into a half-integral
one. This can be done as follows. Observe that adding constraints $z_i=1$
and $z_i=0$ both create systems which correspond to double path systems
for smaller graphs, in the first case a graph where $v_i$ has been
deleted, in the second case a graph where $v_i$ has been bypassed
(creating an edge $v_pv_q$ of the appropriate label for every 2-edge path
$v_pv_iv_q$ through $v_i$), then deleted. Thus the system retains a
half-integral optimum after the addition of such constraints, and we may
simply iteratively add such constraints that fail to raise the optimal
cost, until it is an optimal solution to set $z_i=\half$ for all remaining variables $z_i$.
%\ynote{Grammatically incorrect}.   
\maybeqed{} \end{proof}
%\qed

%
%
%\begin{lemma}
%\label{lemma:double-path-all}
%The double path system has a polynomial-time separation oracle.
%Furthermore it has a half-integral optimum, which can be found in
%polynomial time, which can be converted into an optimal solution for
%\textsc{Half-Integral GFVS with Assignments}. 
%\end{lemma}
%
%Theorem~\ref{theorem:gfvs} follows by combining Lemmas~\ref{lemma:gfvs:invoke} and~\ref{lemma:double-path-all}.

\subsection{Implications}
\label{section:gfvs-implies}

Theorem~\ref{theorem:gfvs} provides the first single-exponential time
algorithm for both \textsc{Group Feedback Vertex Set} and \textsc{Group
  Feedback Edge Set}, with a quite competitive running time;
the existence of such an algorithm was an open question in~\cite{CyganPP12}. 
Via a reduction given in~\cite{CyganPP12}, we furthermore get an algorithm
with the same running time for \textsc{Subset Feedback Vertex Set}, which
was also a previously stated open problem~\cite{CyganPP12}.
%this problem was given an $O^*(2^{O(k\log k)})$-time algorithm
%in~\cite{CyganPPW13}. 

We also observe that, e.g.\ via a group $\Z_2^m$, we can reduce
the basic problem \textsc{Feedback Vertex Set} to GFVS.%
\footnote{We encourage the interested reader to investigate the question
  of how large the group $\Gamma$ needs to be to encode FVS in GFVS. In
  other words, what is the smallest group $\Gamma$ with which you can
  label the edges of $K_n$ so that every simple cycle becomes non-null?
  Our best upper and lower bounds are $O(n^n)$ and $\Omega(n)$,
  respectively (although stronger lower bounds hold for Abelian groups).
  Note that many natural suggestions %such as using $\Z_{n+1}$
  %and labelling every edge by $+1$ 
  fail since labels are direction-dependent.}
While this problem already has faster FPT algorithms (e.g., time
$O^*(3^k)$ by the recent cut-and-count technique~\cite{CyganNPPRW11}),
this is the first LP-branching algorithm for the problem, which may be of
interest by itself (although the LP-formulation is admittedly somewhat
obscure). Our algorithm also distinguishes itself from previous work in
that it never uses the iterative compression technique. 

Furthermore, we observe by completeness that for an explicitly given group
$\Gamma$, we can add the soft versions of constraints $(u=a \lor v=b)$ where $a, b \in
\Gamma$ to the repertoire, and still get a single-exponential running time
(say, $O^*(3^{2k})$ with a rough analysis). Similarly to as in
Section~\ref{section:improved}, we can for each such constraint simply branch
on the cases $(u=a)$, $(v=b)$ and the case that the constraint is false
(details omitted). This may be of interest for the general question of
which VCSPs admit single-exponential time FPT algorithms.

Finally, regarding the use of gap parameters, we note that while GFVS
in ``pure'' form always has a feasible all-relaxed solution of cost zero,
the problem \textsc{GFVS with Assignments} has a relaxation lower bound
which is at least as large as the packing number for paths inconsistent
with the assignments. In particular, when modelling \textsc{Multiway Cut},
this number equals the Mader-path packing number (see~\cite{CyganPPW13MWC}),
and thus the above algorithm, applied to \textsc{Multiway Cut}, is $O^*(2^k)$
(as in~\cite{CyganPPW13MWC}). 
Similar statements can be made about FVS: if $v$ is a vertex of a graph $G$
for which it has been decided that $v$ is not to be deleted, then (but only then)
we may use as a lower bound the ``$v$-flower number'', i.e., the maximum 
number of circuits one can pack, each incident on $v$ 
but otherwise pairwise disjoint.

\section{Linear-time FPT algorithms}
In the previous sections, we have shown that if a problem can be relaxed to basic
$k$-submodular functions, then it can be solved in FPT time.
In this section, we show that, if a problem admits a \textit{binary} basic
$k$-submodular relaxation, then it can be solved in \textit{linear-time} FPT by
computing a network flow and exploiting the structure of the
minimum cuts.

Let $D=\{1,2,\ldots,k\}$ be a domain and $D'=\{0\}\cup D$ be the relaxed domain.
We say that a minimum solution $x\in D'^X$ of a function $f':D'^X\rightarrow \mathbb{R}$ is \textit{dominated} by a
minimum solution $y\in D'^X$ if $x\neq y$ and for any $i\in X$ it holds
that $x_i\neq 0\Rightarrow x_i=y_i$.
If there are no such $y$, we say that $x$ is an \textit{extreme} minimum solution.
In what follows, we prove the following theorem.
\begin{theorem}\label{thm:linear}
Let $f':D'^X\rightarrow\mathbb{N}$ be a sum of $m$ binary basic
$k$-submodular functions.
Then, we can compute an extreme minimum solution of $f'$ in $O((\min f') k m)$ time.
\end{theorem}

Let $x^*$ be the obtained extreme minimum solution of the function $f'$.
Then, for any variable $v\in X$ such that $x^*_v=0$ and for any value $i\in
D$, fixing $x_v$ to $i$ together with the integral part of $x^*$ strictly
increases the optimal value of $f'$.
Thus Theorem~\ref{thm:linear} implies the following corollary.
\begin{corollary}\label{cor:linear}
If a function $f$ can be relaxed to a sum of $m$ binary basic $k$-submodular
functions $f'$, then it can be minimised in $O(k^{2(\min f-\min f')+1} m+(\min
f')km)$ time.
\end{corollary}

Here, a naive algorithm takes $O(k^{2(\min f-\min f')+1}(\min f) m)$ time because it takes $O((\min f)k
m)$ time to compute an extreme minimum solution on each branching node.
However, we can easily separate the coefficient of $\min f$ because we can reuse the
previous minimum solution before a branching to recompute the new minimum
solution after the branching by searching augmenting paths of a network.
Since this optimisation is not important to achieve linear-time complexity, we omit the detail here and refer to
\cite{Iwata14} for a detail discussion.

As we have seen in Section~\ref{section:discrete-relaxation}, both clause-deletion and variable-deletion versions of 
\textsc{Almost 2-SAT} admit binary basic bisubmodular relaxations.
Thus Corollary~\ref{cor:linear} implies that they can be
solved in $O(4^k m)$ time where $m$ is the number of clauses
(as was also shown in~\cite{Iwata14}).
Moreover, as we have seen in Section~\ref{section:ksubmod}, edge-deletion \textsc{Unique Label Cover} admits
a binary basic $|\Sigma|$-submodular relaxation.
Thus it can be solved in $O(|\Sigma|^{2p}m)$ time where $m$ is the number of
edges.

In order to prove Theorem~\ref{thm:linear}, we first introduce some definitions.
For a directed graph $G=(V,E)$ and its vertex subset $S\subseteq V$, we denote the edges outgoing from $S$ by
$\delta^+(S)$ and the edges incoming to $S$ by $\delta^-(S)$.
When $S$ is a single-element set $\{v\}$, we write $\delta^+(v)$ and $\delta^-(v)$, respectively.
For a function $f:U\rightarrow\mathbb{R}$, we denote the sum of $f(a)$ over $a\in S\subseteq U$ by $f(S)=\sum_{a\in
S}f(a)$.
A vertex set $S\subseteq V$ is called \textit{closed} if $\delta^+(S)$ is an
empty set.
A vertex set $S\subseteq V$ is called \textit{strongly connected} if for any two
vertices $u,v\in S$, there is an directed path from $u$ to $v$ in $S$.
It is known that we can compute strongly connected components in $O(|V|+|E|)$
time.
We call a strongly connected component by an \textit{scc} for short.

A \textit{network} is a pair $(G,c)$ of a directed graph $G=(V,E)$ and a capacity function
$c:E\rightarrow\mathbb{R}_{\geq 0}$.
For $s,t\in V$, an \textit{$s$-$t$ flow} of amount $M$ is a function $f:E\rightarrow\mathbb{R}_{\geq 0}$
that satisfies $f(e)\leq c(e)$ for any $e\in E$ and
\begin{align*}
f(\delta^+(v))-f(\delta^-(v))&=\begin{cases}
M & \text{for } v=s,\\
-M & \text{for } v=t,\\
0 & \text{for any } v\in V\setminus\{s,t\}.
\end{cases}
\end{align*}
For convenience, we define $c(e)=f(e)=0$ if $e\not\in E$.
A vertex subset $S$ is called an $s$-$t$ cut if $s\in S$ and $t\not\in S$, and its \textit{capacity} is defined as
$c(S)=c(\delta^+(S))$.
The \textit{residual graph} of a network $(G,c)$ with respect to a flow $f$ is the
directed graph $G_f=(V,E_f)$ with $E_f=\{(u,v) \mid f(u,v)<c(u,v) \text{ or }
f(v,u)>0\}$.

Let $f:D'^X\rightarrow\mathbb{R}$ be a function on a domain $D'=\{0,1,2,\ldots,k\}$.
Now, we aim to express $f$ as cuts of a network.
For a variable $v\in X$, we denote a vertex set $\{v_i\mid i\in D\}$ by $X_v$.
An \textit{$(X,k)$-network} is a network on vertices $V=\bigcup_{v\in X}X_v\cup \{s,t\}$.
For an assignment $\phi:X\rightarrow D'$, we define the \textit{$s$-$t$ cut
corresponding to $\phi$} as the set of vertices consisting of $v_{\phi(v)}$ for
each variable $v\in X$ such that $\phi(v)\neq 0$ together with $s$, which is denoted as $S_\phi$.
That is, $S_\phi=\{s\}\cup \{v_{\phi(v)}\mid v\in X, \phi(v)\neq 0\}$.
If an $s$-$t$ cut contains at most one vertex from each $X_v$, it is called \textit{normalised}.
Note that $S_\phi$ is a normalised cut for any $\phi$.
For a normalised cut $S$, we define the \textit{assignment corresponding to $S$}
as $\phi_S(v)=i$ if $S\cap X_v=\{v_i\}$ and $\phi_S(v)=0$ if $S\cap X_v=\emptyset$.
We say that an $(X,k)$-network \textit{represents} $f$ if for any
assignment $\phi:X\rightarrow D'$, the capacity of the corresponding cut $S_\phi$ is equal to the value of the function $f(\phi)$.
We say that a function $f$ is \textit{representable} if there is an
$(X,k)$-network that represents $f$.
For an $s$-$t$ cut $S\subseteq V$, we define the \textit{normalised cut of $S$},
which is denoted by $\nu(S)$, as the set of vertices consisting of $S\cap X_v$ for each variable $v\in X$ such that $|S\cap X_v|=1$ together with $s$.
That is, $\nu(S)=\{s\}\cup\{v_i\mid v\in X, S\cap X_v=\{v_i\}\}$.
We say that an $(X,k)$-network is \textit{$k$-submodular} if for
any $s$-$t$ cut $S$, it holds that $c(S)\geq c(\nu(S))$.
If there is a $k$-submodular $(X,k)$-network that represents a function $f$, we
say that $f$ is \textit{$k$-submodular representable}.
A normalised minimum cut $S$ is called \textit{dominated} by a normalised
minimum cut $S'$ if it holds that $S\subset S'$.
If there are no such $S'$, we say that $S$ is an \textit{extreme} minimum cut.

\begin{lemma}
Let $f:D'^X\rightarrow\mathbb{R}$ be a sum of functions $f_1,\ldots,f_m$.
If for each summand function $f_i$ on variables $Y_i\subseteq X$, there exists an
$(Y_i,k)$-network $(G_i=(\bigcup_{v\in Y_i}X_v\cup\{s,t\},E_i),c_i)$,
then their sum $(G=(\bigcup_{v\in X}X_v\cup\{s,t\}, \bigcup_{i=1}^mE_i), \sum_{i=1}^m c_i)$ is an
$(X,k)$-network that represents $f$.
If each network is $k$-submodular, then the sum of the networks is also $k$-submodular.
\end{lemma}
\begin{proof}
Trivial because the capacity of the cut on $\sum_{i=1}^m c_i$ is equal to the sum of the capacities of the cut on each
$c_i$.
\end{proof}

\begin{lemma}
If a function $f$ is $k$-submodular representable, then $f$ can be minimised
by computing the minimum $s$-$t$ cut of the network.
\end{lemma}
\begin{proof}
Since the network represents $f$, for any assignment $\phi$, it holds that $c(S_\phi)=f(\phi)$.
Let $\phi$ be a minimiser of $f$, and let $S$ be a minimum $s$-$t$ cut of the network.
Because the network is $k$-submodular, $\nu(S)$ is also a minimum $s$-$t$ cut.
Therefore, $f(\phi_{\nu(S)})=c(\nu(S))\leq c(S_\phi)=f(\phi)$ holds.
Since $\phi$ is a minimiser of $f$, $\phi_{\nu(S)}$ is also a minimiser of $f$.
\end{proof}

In order to obtain an extreme minimum solution, we prove the following
one-to-one correspondence between the extreme minimum solution and the extreme minimum cut.
\begin{lemma}\label{lem:corextreme}
Let $f:D'^X\rightarrow\mathbb{R}$ be a function and $(G,c)$ be a
$k$-submodular network that represents $f$.
Then, an assignment $\phi: X\rightarrow D'$ is an extreme minimum
solution if and only if its corresponding cut $S_\phi$ is an extreme minimum
cut.
\end{lemma}
\begin{proof}
$(\Rightarrow)$ Let $S$ be a normalised minimum cut.
If there exists a normalised minimum cut $S'$ that dominates $S$, then, from the
definition, it holds that $\phi_S\neq \phi_{S'}$ and $\phi_{S}(v)\neq 0\Rightarrow\phi_{S}(v)=\phi_{S'}(v)$.
Thus, $\phi_S$ is not an extreme minimum solution.

$(\Leftarrow)$ Let $\phi$ be a minimum solution.
If there exists a minimum solution $\phi'$ that dominates $\phi$,
then, from the definition, it holds that $S_{\phi}\subset S_{\phi'}$.
Thus, $S_{\phi}$ is not an extreme minimum cut.
\end{proof}

From the above lemma, in order to compute an extreme minimum solution, it suffices
to compute an extreme minimum cut.
In order to compute an extreme minimum cut, we introduce the following one-to-one
correspondence between the minimum $s$-$t$ cut and the closed vertex set of the
residual graph.
\begin{lemma}[Picard and Queyranne~\cite{PQ80}]\label{lem:pq}
For any network, its two vertices $s,t$, and its maximum $s$-$t$ flow $f$, an
$s$-$t$ cut $S$ is a minimum cut if and only if $S$ is a closed set containing $s$ in
the residual graph with respect to $f$.
\end{lemma}

Note that a maximum $s$-$t$ flow in the lemma is arbitrary.
This lemma reveals a nice structure of the all minimum cuts:
although there exist exponentially many minimum cuts in a network, we can find
an extreme one in linear-time as the following lemma.

\begin{lemma}\label{lem:algextreme}
Let $(G,c)$ be a $k$-submodular $(X,k)$-network and $f$ be a maximum $s$-$t$ flow of the
network.
Then, an extreme minimum cut of the network can be computed in $O(|V|+|E|)$
time.
\end{lemma}
\begin{proof}
The algorithm is described in Algorithm~\ref{alg:extreme}.
First, we compute the strongly connected components of the residual graph $G_f$.
From Lemma~\ref{lem:pq}, for each strongly connected component $T$, any minimum
cut must contain all of $T$ or none of $T$.
Then we compute the vertex set $S$ reachable from $s$ in $G_f$.
Since this is a closed set containing $s$, it is a minimum cut.
Suppose that $S$ is not a normalised cut.
Since the network is $k$-submodular, $\nu(S)\subset S$ is also a minimum cut.
From Lemma~\ref{lem:pq}, this means that there are no outgoing edges from $\nu(S)$ in $G_f$, which contradicts the fact
that $S$ is the set reachable from $s$.
Thus, $S$ is a normalised minimum cut.
From now on, we modify $S$ to be an extreme minimum cut by
expanding it.
Let $T\subseteq V\setminus S$ be a strongly connected component that satisfies
the following two conditions:
\begin{enumerate}
  \item All the outgoing edges from $T$ are coming into $S$.
  \item The cut $S\cup T$ is normalised.
\end{enumerate}
If there exists a strongly connected component $T$ that satisfies the first condition, the cut $S\cup T$ also
becomes a closed set.
Thus it is a minimum cut.
If there exists $T$ that satisfies both of the conditions,
we can obtain a new normalised minimum cut by expanding $S$ to $S\cup T$.
If there are no such $T$, $S$ is an extreme cut.
This is because any minimum cut $S'\supset S$ must contain at least one of the
strongly connected components that satisfy the condition 1, but including any
of them does not lead to a normalised cut.

Finally, we analyze the running time of the algorithm.
We can compute the strongly connected components in $O(|V|+|E|)$ time.
In order to efficiently find a strongly connected component that satisfies the condition 1, for each strongly
connected component $T$, we keep track of the number of edges outgoing from $T$ to the vertices outside $S$.
If this number is zero, it satisfies the condition 1.
When updating $S$ to $S\cup T$, for each edge $uv\in\delta^-(T)$, we decrement the number for the strongly connected
component that contains $u$.
This takes only $O(|\delta^-(T)|)$ time for each $T$.
Thus it takes only $O(|E|)$ time in total.
% By managing the number of edges outgoing from each
% strongly connected component to vertices outside $S$, we can obtain the list of
% the strongly connected components that satisfy the condition 1.\footnote{Magnus: Shouldn't this be about managing \emph{incoming} edges into $T$ for every scc $T$?}
% When updating $S$ to $S\cup T$, it takes only $O(|\delta^-(T)|)$ time to update
% the list.
% So it takes only $O(|E|)$ time for this part.
If a strongly connected component $T$ does not satisfy the condition 2 for some
$S$, it will never satisfy the condition for any $S'\supset S$.
Therefore, we don't have to check the same strongly connected component
multiple times.
Thus the total running time is $O(|V|+|E|)$.
\end{proof}

\begin{algorithm}[tb]
\caption{Algorithm to compute an extreme minimum cut.}
\label{alg:extreme}
\begin{algorithmic}[1]
\INPUT the residual graph $G_f$ of an $(X,k)$-network
\OUTPUT an extreme minimum cut
\State compute the strongly connected components
\State $S\leftarrow$ the vertices reachable from $s$
\While{$\exists$ unchecked scc $T$ such that $\delta^+(T)\subseteq S$}
	\If{$S\cup T$ is a normalised cut}
		\State $S\leftarrow(S\cup T)$
	\EndIf
\EndWhile
\State \Return $S$
\end{algorithmic}
\end{algorithm}

Now we show that any binary basic
$k$-submodular function is $k$-submodular representable.
For the definition of the basic $k$-submodular functions, please refer to Lemma~\ref{lemma:usefulksubmod}.

\begin{figure}[tb]
 \begin{minipage}{0.38\hsize}
  \begin{center}
  \includegraphics[width=\hsize]{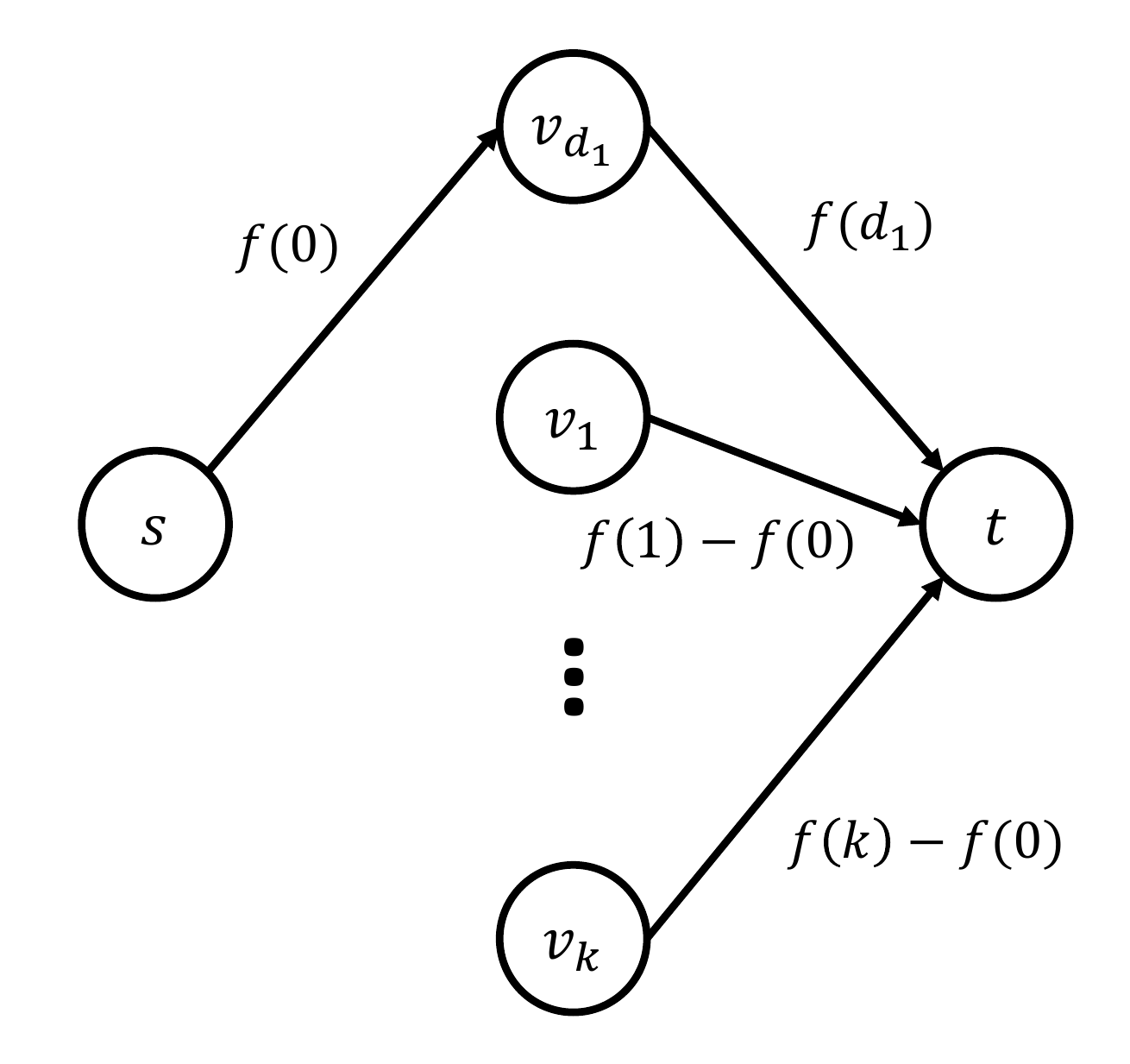}
  \end{center}
  \caption{Unary $f(v)$}
  \label{fig:unary}
 \end{minipage}
 \begin{minipage}{0.3\hsize}
  \begin{center}
  \includegraphics[width=\hsize]{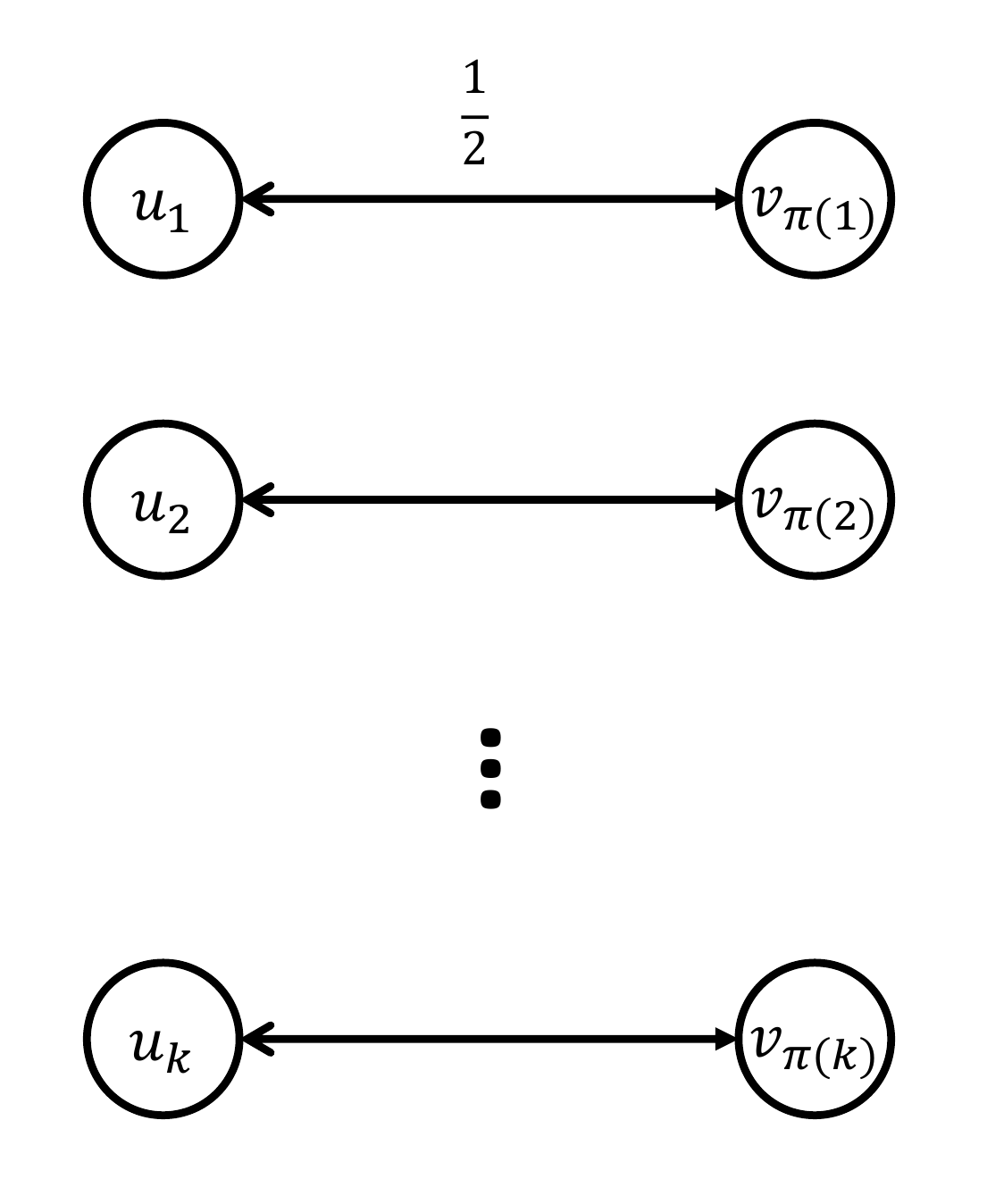}
  \end{center}
  \caption{$(v=\pi(u))$}
  \label{fig:perm}
 \end{minipage}
 \begin{minipage}{0.3\hsize}
  \begin{center}
  \includegraphics[width=\hsize]{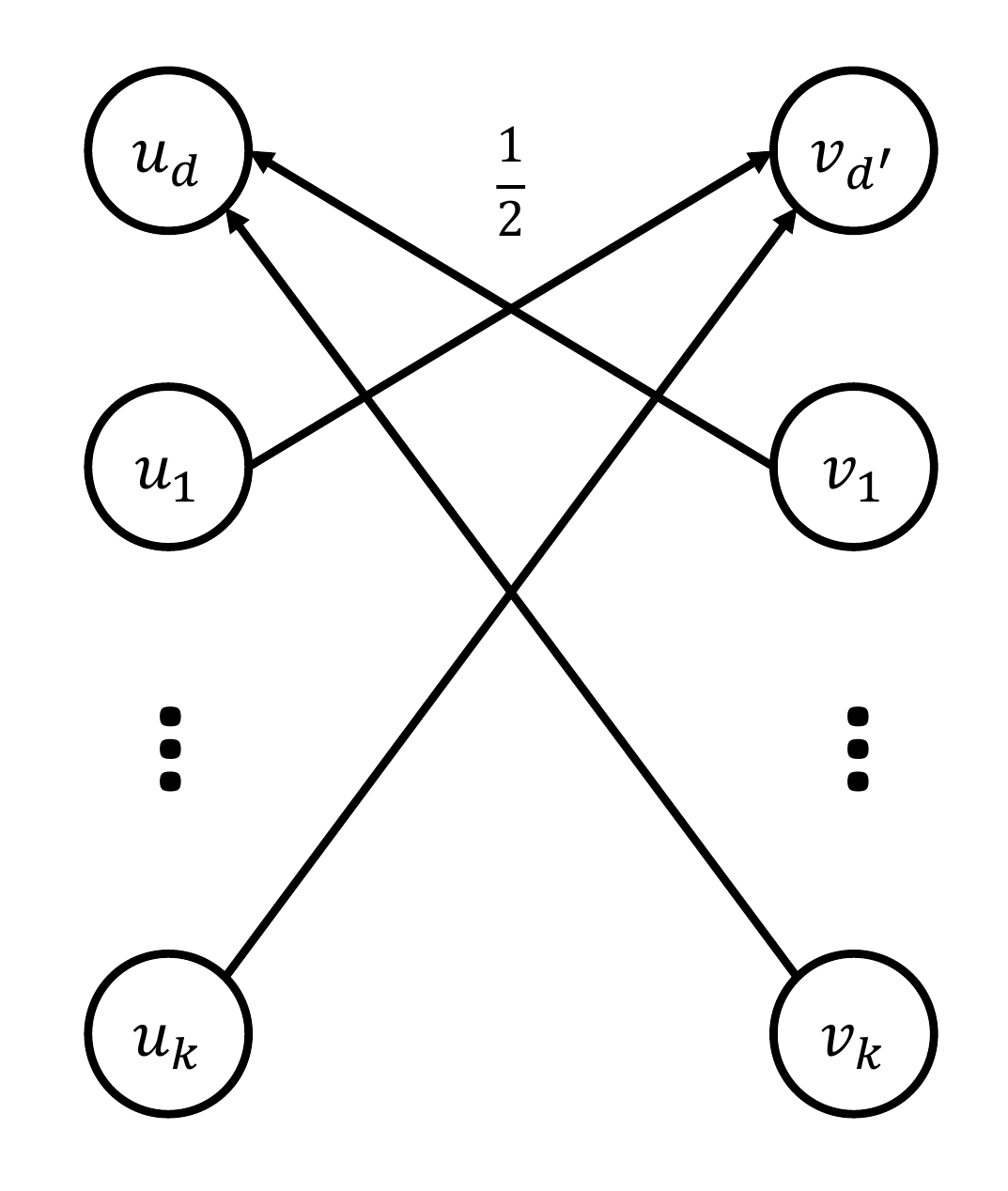}
  \end{center}
  \caption{$(u=d\vee v=d')$}
  \label{fig:or}
 \end{minipage}
\end{figure}

\begin{lemma}\label{lem:netunary}
Any unary function $f:D'\rightarrow\mathbb{R}$ is $k$-submodular
representable.
\end{lemma}
\begin{proof}
By subtracting the minimum value of $f$, we can assume that $f$ is nonnegative.
Let $d_1=\arg\min_{d\in D} f(x)$.
Then, we construct a $(\{v\},k)$-network as follows (Figure~\ref{fig:unary}):
\begin{itemize}
  \item $c(s,v_{d_1})=f(0)$,
  \item $c(v_{d_1},t)=f(d_1)$,
  \item $c(v_d,t)=f(d)-f(0)$ for any $d\neq d_1$.
\end{itemize}
Note that, for $d\neq d_1$, $f(d)-f(0)\geq 0$ holds because it holds that
$2f(0)\leq f(d_1)+f(d)\leq 2f(d)$.

If $\phi(v)=0$, the capacity of the
corresponding cut is $c(S_\phi)=c(s,v_{d_1})=f(0)$.
If $\phi(v)=d_1$, the capacity of the corresponding cut is
$c(S_\phi)=c(v_{d_1},t)=f(d_1)$.
If $\phi(v)=d$ for $d\neq d_1$, the capacity of the corresponding cut is
$c(S_\phi)=c(s,v_{d_1})+c(v_d,t)=f(d)$.
Thus the network actually represents $f$.

Let $D'\subseteq D$ be a set of size at least $2$ and let
$S=\{s\}\cup\{v_d\mid d\in D'\}$ be a cut.
When $D'$ does not contain $d_1$, let $d_2,d_3$ be distinct elements contained
in $D'$.
Then, $c(S)$ is at least
$c(s,v_{d_1})+c(v_{d_2},t)+c(v_{d_3},t)=f(d_2)+f(d_3)-f(0)$.
Since $f$ is $k$-submodular, $f(d_2)+f(d_3)\geq 2f(0)$.
Therefore, $c(S)\geq f(0)=c(\nu(S))$ holds.
When $D'$ contains $d_1$, let $d_2$ be another element contained in $D'$.
Then, $c(S)$ is at least
$c(v_{d_1},t)+c(v_{d_2},t)=f(d_1)+f(d_2)-f(0)\geq f(0)$.
Therefore, $c(S)\geq c(\nu(S))$ holds.
Thus the network is actually $k$-submodular.
\end{proof}

\begin{lemma}\label{lem:netperm}
For any permutation $\pi$ on $D$, the basic $k$-submodular relaxation $f$ of the
soft version of a constraint $(x=\pi(y))$ is $k$-submodular representable.
\end{lemma}
\begin{proof}
Let $u,v$ be variables.
We construct a $(\{u,v\},k)$-network as follows (Figure~\ref{fig:perm}):
\begin{itemize}
  \item $c(u_i,v_{\pi(i)})=\frac{1}{2}$ for any $i\in D$,
  \item $c(v_j,u_{\pi^{-1}(j)})=\frac{1}{2}$ for any $j\in D$.
\end{itemize}

If $\phi(u)=\phi(v)=0$, the capacity of the corresponding
cut is $c(S_\phi)=0=f(\phi)$.
If $\phi(u)=i\in D$ and $\phi(v)=0$, the capacity of the corresponding cut is
$c(S_\phi)=c(u_i,v_{\pi(i)})=\frac{1}{2}=f(\phi)$.
Similarly, if $\phi(u)=0$ and $\phi(v)\neq 0$, the capacity of the corresponding
cut is equal to $f(\phi)$.
If $\phi(u)=i\in D, \phi(v)=j\in D$ and $j=\pi(i)$, the capacity of the
corresponding cut is $c(S_\phi)=0=f(\phi)$.
Otherwise, the capacity of the corresponding cut is
$c(S_\phi)=c(u_i,v_{\pi(i)})+c(v_j,u_{\pi^{-1}(j)})=1=f(\phi)$.
Thus the network actually represents $f$.

Let $S$ be a cut and $I,J$ be two sets such that $I=\{i\in D\mid u_i\in S\}$
and $J=\{j\in D\mid v_j\in S\}$.
If $|I|\leq 1$ and $|J|\leq 1$, the cut $S$ is already normalised.
If $|I|=0$ or $|I|\geq 2$, and $|J|=0$ or $|J|\geq 2$, the capacity of the
normalised cut is $c(\nu(S))=c(\{s\})=0$ and the capacity of the
original cut is nonnegative. Therefore, $c(S)\geq c(\nu(S))$ holds.
If $I=\{i\}$ and $|J|\geq 2$, the capacity of the normalised cut is
$c(\nu(S))=c(\{s,u_i\})=c(u_i,v_{\pi(i)})=\frac{1}{2}$.
Because $\pi$ is a permutation, for at least one $j\in J$, $\pi^{-1}(j)$ is
different from $i$.
Therefore, the capacity of the original cut is at least $\frac{1}{2}$.
Thus, it holds that $c(S)\geq c(\nu(S))$.
Similarly, if $|I|\geq 2$ and $|J|=1$, it holds that $c(S)\geq c(\nu(S))$.
Thus, the network is actually $k$-submodular.
\end{proof}

\begin{lemma}\label{lem:netor}
For any $d,d'\in D$, the basic $k$-submodular relaxation $f$ of the soft version
of a constraint $(x=d\vee y=d')$ is $k$-submodular representable.
\end{lemma}
\begin{proof}
Let $u,v$ be variables.
We construct a $(\{u,v\},k)$-network as follows (Figure~\ref{fig:or}):
\begin{itemize}
  \item $c(u_i,v_{d'})=\frac{1}{2}$ for any $i\in D\setminus\{d\}$,
  \item $c(v_j,u_{d})=\frac{1}{2}$ for any $j\in D\setminus\{d'\}$.
\end{itemize}

If $\phi(u)=\phi(v)=0$, $\phi(u)=d$, or $\phi(v)=d'$, the capacity of the
corresponding cut is $c(S_\phi)=0=f(\phi)$.
If $\phi(u)=i\in D\setminus\{d\}$ and $\phi(v)=0$, the capacity of the
corresponding cut is $c(S_\phi)=c(u_i,v_d')=\frac{1}{2}=f(\phi)$.
Similarly, if $\phi(u)=0$ and $\phi(v)\in D\setminus\{d'\}$, the capacity of
the corresponding cut is equal to $f(\phi)$.
If $\phi(u)=i\in D\setminus\{d\}, \phi(v)=j\in D\setminus\{d'\}$, the
capacity of the corresponding cut is
$c(S_\phi)=c(u_i,v_{d'})+c(v_j,u_d)=1=f(\phi)$.
Thus the network actually represents $f$.

Let $S$ be a cut and $I,J$ be two sets such that $I=\{i\in D\mid u_i\in S\}$
and $J=\{j\in D\mid v_j\in S\}$.
If $|I|\leq 1$ and $|J|\leq 1$, the cut $S$ is already normalised.
If $|I|=0$ or $|I|\geq 2$, and $|J|=0$ or $|J|\geq 2$, the capacity of the
normalised cut is $c(\nu(S))=c(\{s\})=0$ and the capacity of the
original cut is nonnegative. Therefore, $c(S)\geq c(\nu(S))$ holds.
If $I=\{d\}$ and $|J|\geq 2$, both of the normalised cut and the original cut
have the capacity zero.
If $I=\{i\}$ for $i\neq d$ and $|J|\geq 2$, since $J$ contains at least one
element $j$ which is different from $d'$, the capacity of the original cut is at
least $c(v_j,u_d)=\frac{1}{2}$.
On the other hand, the capacity of the normalised cut is
$c(\nu(S))=c(u_i,v_d')=\frac{1}{2}$.
Therefore, it holds that $c(S)\geq c(\nu(S))$.
Similarly, if $|I|\geq 2$ and $|J|=1$, it holds that $c(S)\geq c(\nu(S))$.
Thus, the network is actually $k$-submodular.
\end{proof}

Finally, we prove Theorem~\ref{thm:linear}.
\begin{proof}[Proof of Theorem~\ref{thm:linear}]
By using Lemmas~\ref{lem:netunary}--\ref{lem:netor}, 
we can construct a $k$-submodular $(X,k)$-network $(G,c)$ that represents $f$
in $O(|G|)$ time.
Since we create $O(k)$ edges per each summand function $f_i$, the size of the
network is $O(km)$.
Because the capacity of the minimum cut of the network is equal to $\min f$ and
each capacity is a multiple of $\frac{1}{2}$, we can compute the maximum flow of
the network in $O((\min f) k m)$ time.
Then, by using Lemma~\ref{lem:algextreme}, we can compute an extreme minimum cut
in $O(km)$ time.
Finally, by using Lemma~\ref{lem:corextreme}, we can obtain an extreme minimum
solution.
The total running time is $O((\min f)km)$.
\end{proof}

\section{Conclusions and open problems}

We have shown that half-integrality and LP-branching can be powerful
tools for FPT-algorithms, beyond just \textsc{Vertex Cover} and
\textsc{Multiway Cut}. We have outlined how to use CSP tools to find and
study such relaxations. As an application, we have given new half-integral
relaxations for \textsc{Unique Label Cover} and \textsc{Group Feedback
  Vertex Set}, in both cases improving the running time asymptotically
(to single-exponential for fixed label set, resp.\ to unconditionally single-exponential).
Several directions of further study suggest themselves. 
Is there a way to decide the existence of discrete relaxations in general? 
Can directed problems, e.g., \textsc{Directed Feedback Vertex Set} be
handled in a similar manner? Finally, can the basic tool of LP-branching
be complemented with more sophisticated algorithmic approaches (e.g.,
FPT-time separation oracles, or tools from semi-definite programming)?

In other directions, we note that several of the covered problems have
polynomial kernels for specific cases, e.g., \textsc{Group Feedback Vertex
  Set} with bounded-size group~\cite{KratschW12b} and \textsc{Feedback
  Vertex Set}~\cite{Thomasse10}; it is an interesting question how far
this can be generalised.

We also note that oracle minimisation of $k$-submodular functions is an
open question; we also welcome more investigation into $k$-submodular
functions in general (including, e.g., any possible connections to 
path-packing systems as in~\cite{ChudnovskyCG08,ChudnovskyGGGLS06,Pap07,Pap08}, 
and algebraic algorithms generalising those for matching; see also~\cite{Yamaguchi14}). 

As for linear-time complexity, we have shown that edge-deletion \textsc{Unique Label Cover} can be solved in
linear-time.
It is known that \textsc{Multiway Cut}, a special case of \textsc{Unique Label Cover}, can be solved in linear-time even
for the node-deletion version~\cite{Iwata14}.
It is an interesting question whether node-deletion \textsc{Unique Label Cover} can also be solved in linear-time.
In order to obtain linear-time FPT algorithms, we have shown that we can minimize a sum of \textit{basic} binary
$k$-submodular functions via network flow.
We left whether it is possible to minimize a sum of \textit{any} binary $k$-submodular
functions in a similar way or not as an open problem.

\section*{Acknowledgements} 
The second author thanks Marek Cygan, 
Andreas Karrenbauer, Johan Thapper and Stanislav \v{Z}ivn\'y for enlightening discussions. 

\bibliographystyle{abbrv}
\bibliography{ilp-fpt}

%\newpage
\appendix

\section{On the crisp solution structure supported by the algorithms}

Now, we discuss the crisp solution structure supported by bisubmodular and $k$-submodular functions
(in particular, we prove Lemma~\ref{lemma:ksubstructure}). 

To illustrate the topic, let us focus on the (well-understood) case of submodular functions.
It is known that for a submodular function $f: 2^V \rightarrow \R$, one can not only minimise $f(S)$
efficiently in an unconstrained setting, but also subject to a \emph{ring family}. 
Recall that a ring family is a set family ${\cal{F}} \subseteq 2^V$ which is closed under
union and intersection, i.e., if $A, B \in {\cal{F}}$ then $A \cup B, A \cap B \in \cal{F}$.
The constrained optimisation problem is then
$
\min_{S \in {\cal{F}}} f(S),
$
which can be solved in polynomial time even if $f$ is only given via oracle access (see Schrijver~\cite{SchrijverBook}). 

Now observe that the conditions on a ring family are actually \emph{polymorphisms} of the relation $R(S)=(S \in \cal{F})$. 
Indeed, it is known that a relation $R \subseteq 2^V$ is closed under union and intersection
if and only if $R$ can be modelled as the set of solutions to a formula using constraints $(x \rightarrow y)$, 
$(x=0)$, and $(x=1)$ (e.g., the set of closed vertex sets in a digraph). % \ynote{What does this mean?}).
Furthermore, if $f$ is a submodular function, then the set of minimising assignments 
${\cal{A}}=\{A \subseteq V: f(A)=\min_S f(S)\}$ is itself closed under union and intersection 
(by applying the submodularity condition $f(A)+f(B) \geq f(A \cap B)+f(A \cup B)$ to two minimising assignments
$A, B \in {\cal{A}}$). 
Thus, if we want to \emph{implement} some crisp solution structure on the search space $2^V$ by only using
the power of submodular functions, then this restriction must take the shape of a ring family, 
and if it does, then it is sufficient to implement the crisp constraints $(x \rightarrow y), (x=0)$, and $(x=1)$,
which can be done by using their soft versions at very high cost; these soft versions are submodular,
which closes the loop.

Expressed more succinctly, if one wants to perform constrained minimisation of a submodular function 
without using any algorithm more powerful than basic (unconstrained) submodular minimisation, 
then the power one has at hand is exactly that of crisp implications and assignments. 
We will investigate the same for functions with bisubmodular or $k$-submodular relaxations. 
Let us finally remark that this is not a restriction on submodular functions themselves; 
submodular functions in general are far more expressive than digraph cut functions
(this has been proven formally in~\cite{ZivnyCJ09}). 

\subsection{Bisubmodular relaxations}

We now consider the bisubmodular case of the above, i.e.,
relaxations of functions $f_i: 2^V \rightarrow \R$
into bisubmodular functions $f_i': \{0,\half,1\}^V \rightarrow \R$.
We consider the structure of the minimising set ${\cal{A}}$ when
restricted to integral assignments (i.e., those half-integral
minimisers of $f'$ which happen to also be integral; 
note that this may well be an empty set). 
We find that \textsc{Bisubmodular Cost 2-SAT} exactly captures its structure. 

\begin{lemma}
\label{lemma:bisub-2cnf}
Let $f: \{0,\half,1\}^V \rightarrow \R$ be a bisubmodular function, and
${\cal{A}} \subseteq \{0,\half,1\}^V$ be its set of minimising assignments.
Then the \emph{integral global minimisers} ${\cal{A}} \cap \{0,1\}^V$ of $f$
can be modelled as the set of solutions to a (crisp) 2-CNF formula $F$
on~$V$. 
\end{lemma}
\begin{proof}
Let ${\cal{A}}_{01} = {\cal{A}} \cap \{0,1\}^V$. We will show that ${\cal{A}}_{01}$
can be described by a 2-CNF formula. As discussed above for the submodular case,
${\cal{A}}$ as a whole must be closed under the operations $\sqcap$ and $\sqcup$,
i.e., $\sqcup$ and $\sqcap$ are polymorphisms of ${\cal{A}}$. 
Define $h(A,B,C)=(((A \sqcap B) \sqcup (A \sqcap C)) \sqcup (B \sqcap C))$; 
then $h$ is a ternary polymorphism of ${\cal{A}}$, and it can be verified that $h$ 
is a majority operation. Thus ${\cal{A}}$ is fully described by the binary constraints 
that it implies (see preliminaries). In turn, each binary constraint
$R(x,y)$ can of course be described by enumerating the forbidden values
of the pair $(x,y)$. % \ynote{I think this is a typo. Should be ``the forbidden values of the pair $(x,y)$?''}. 
Thus, for every point in $\phi\in\{0,1\}^n$ which is not a
point of ${\cal{A}}_{01}$, there is a binary constraint $R(x,y)$ which rejects it. 
All such binary constraints on $\{0,1\}$ can be described via 2-clauses. 
\maybeqed \end{proof}

%(Note that this paragraph only holds for instances with a relaxation gap
%of zero. In general, the solution set to a \textsc{Bisubmodular 2-SAT}
%formula could have an arbitrary structure. However, it does imply that
%if one wants to impose a non-binary constraint $R$ on the structure of the
%solution set, then one must pay for $R$ via the relaxation gap.)

\subsection{$k$-Submodular relaxations}

For $k>2$, the situation is more complicated than above. 
The setup is the same: if ${\cal{A}} \subseteq \{0, \ldots, k\}^V$
is the set of minimising assignments to a $k$-submodular function $f$,
then we look at the structure of the subset ${\cal{A}}_{\mathrm{int}}=
{\cal{A}} \cap \{1, \ldots, k\}^V$ of those assignments which are also integral.
As before, the structure can be defined by a formula over binary (crisp) 
constraints, however, the set of binary constraints we can use is limited.
As stated in Lemma~\ref{lemma:ksubstructure}, it turns out that the binary
constraints of Lemma~\ref{lemma:usefulksubmod} is exactly the right list.

%\begin{lemma}
%\label{lemma:ksubstructure}
%Let $f$ be a $k$-submodular function on $D^n$, and let $P \subseteq D^n$
%be the set of points $X$ that minimise $f(X)$. Let $P_{\mathrm{int}}=P \cap
%\{1,\ldots,k\}^n$. Then $P_{\mathrm{int}}$ can be described as the set of
%solutions to a formula over arbitrary unary constraints and constraints
%$(x=a \lor y=b)$ and $(x=\pi(y))$ (defined as in Lemma~\ref{lemma:usefulksubmod}). 
%\end{lemma}

\begin{proof}[Proof of Lemma~\ref{lemma:ksubstructure}.]
To begin with, we observe as in the proof of Lemma~\ref{lemma:bisub-2cnf}
that binary (and unary) constraints must suffice to describe the
structure. In fact, the same construction of a majority polymorphism
$h(A,B,C)$ from $\sqcap$ and $\sqcup$ applies directly for $k>2$,
hence ${\cal{A}}$, and by implication ${\cal{A}}_{\mathrm{int}}$, 
is fully characterised by its 2-variable projections. 
The remaining task is thus to characterise those crisp binary
constraints on domain $\{1,\ldots,k\}$ whose soft versions have bisubmodular
relaxations. 
%The crux this time
%is that we wish to describe $P_{\mathrm{int}}$ in terms of binary
%constraints which themselves have $k$-submodular relaxations (as valued
%constraints), and this no longer covers all possible binary constraints.  
%
%Thus, our remaining work is to establish the binary
%constraints that are closed under $\sqcap$ and $\sqcup$. 
By Lemma~\ref{lemma:usefulksubmod}, we can support arbitrary unary constraints,
thus we focus on the properly binary constraints.

%The unary constraints we have available are easily characterised: 
%a unary relation $R(x) \subseteq \{0,\ldots,k\}$ is closed under 
%$\sqcup, \sqcap$ if either $R(x)=(x=d)$ for some $d \in \{1,\ldots,k\}$, 
%or $0 \in R$ (with no further restrictions). Thus in particular,
%any unary function on $\{1,\ldots,k\}$ has a $k$-submodular relaxation.

For the rest of the proof, we let $R \subseteq \{0,\ldots,k\}^2$ be a
binary relation closed under $\sqcup$ and $\sqcap$. 
We will characterise the possible sets $R \cap \{1,\ldots,k\}^2$ 
of integral pairs satisfying $R$. 
Let~$S_1=\{a \in \{1,\ldots,k\}: (a,b) \in R \text{ for some } b\}$ 
and $S_2=\{b \in \{1,\ldots,k\}: (a,b) \in R \text{ for some } a\}$
be the integral values that occur in positions $1$ and $2$ of $R$, respectively;
they can be assumed to be non-empty, as otherwise $R$ is simply a
conjunction of an assignment and a unary constraint. 

We begin by a useful property.

\begin{claim}\label{claim:always}
If $(a,0) \in R$ for some $a \in S_1$, then for every $b \in S_2$ we have
$(a,b) \in R$. Thus in particular, for every $a \in S_1$ there is some $b
\in S_2$ such that $(a,b) \in R$. 
\end{claim}
\begin{proof}
If $(0,b) \in R$, then we have $(a,b) \in R$ by $(a,0)\sqcup (0,b)=(a,b)$.

On the other hand, if $(a',b) \in R$ for some $a' \in S_1$
with $a' \neq a$, then $(0,b) \in R$ by
$(a',b) \sqcup (a,0) = (0,b)$, and we are back in the first case.
\maybeqed{} \end{proof}

We eliminate some quick corner cases. Recall that we are focusing on
expressing ${\cal{A}}_{\mathrm{int}}$ via binary relations, rather than all of ${\cal{A}}$; hence
if the intersection of $R$ with $\{1,\ldots,k\}^2$ is simple, we may
ignore complications involving the value $0$. In particular, consider the case that
$|S_1|=1$, say $S_1=\{a\}$. By the above, $(a,b) \in R$ for every $b \in S_2$, 
implying that the effect of $R(x,y)$ on ${\cal{A}}_{\mathrm{int}}$ is simply the
conjunction of $(x=a)$ and $(y \in S_2)$. We claim similarly if $|S_2| = 1$.
Thus in the sequel, we have $|S_1|, |S_2| > 1$. 

We give the next useful observation.

\begin{claim}
\label{claim:oneoroall}
For any $a \in S_1$,
either there is exactly one value $b \in S_2$ such that $(a,b) \in R$,  
or $(a,b) \in R$ for every $b \in S_2$. 
Symmetrically, for any $b \in S_2$,
either there is exactly one value $a \in S_1$ such that $(a,b)\in R$, 
or $(a,b) \in R$ for every $b \in S_1$. 
\end{claim}
\begin{proof}
We prove the claim for some $a \in S_1$; the other half is entirely symmetric. 
Recall that $(a,b) \in R$ for at least one $b \in S_2$, by previous
claims. Thus let $(a,d), (a,d') \in S$ for $d, d' \in S_2$, $d \neq d'$;
this produces $(a,0) \in R$ via the polymorphism $\sqcup$, and by the previous claim
$(a,b) \in R$ for every $b \in S_2$, as claimed. 
\maybeqed{} \end{proof}

We call a value $a \in S_1$ (resp.\ $b \in S_2$) \emph{global} if the
second case occurs, i.e., if $(a,d) \in R$ for every $d \in S_2$
(resp.\ $(d,b) \in R$ for every $d \in S_1$). 
We may assume that each of $S_1$ and $S_2$ contains at most one global
value: if $S_1$ contains two global values $a, a'$, then every value in
$S_2$ must be global, and since $|S_2|>1$ we get that every value in $S_1$
is global, and the effect of $R$ on ${\cal{A}}_{\mathrm{int}}$ can be described via
unary constraints.

Furthermore, if $a \in S_1$ and $b \in S_2$ are global values, then for
any $a' \in S_1$, $a' \neq a$, we have that $(a',b) \in R$ is the unique 
occurrence of $a'$ in $R$; hence the effect of $R(x,y)$ on
${\cal{A}}_{\mathrm{int}}$ can be given as $(x=a \lor y=b)$ in 
conjunction with unary constraints. 
Note that this is case~3 of Lemma~\ref{lemma:usefulksubmod}. 

Second, assume that $S_2$ contains no global values, but $a \in S_1$ is
global. But there is one further $a' \in S_1$, with $(a',b) \in R$ for
some $b \in S_2$ by Claim~\ref{claim:always}; hence $b$ is global and we are back at a previous case. 

Finally, if there are no global values, then the values of $S_1$ and $S_2$
must be matched to each other with exactly one possible value each. 
We may thus describe $R$ as a bijection $(x = \pi(y))$ in conjunction with
a unary constraint, i.e., case~2 of Lemma~\ref{lemma:usefulksubmod}.
This finishes the proof. 
\maybeqed{} \end{proof}
%\qed

Note that this is not a complete characterisation of the full set ${\cal{A}}$
of minimisers, since we skipped some ``corner cases'' that become uninteresting
when intersected with $\{1,\ldots,k\}^V$. Also note, as in the discussion
for submodular functions, that this does not imply that Lemma~\ref{lemma:usefulksubmod} 
can produce all functions with $k$-submodular relaxations, as valued constraints
taking several values (beyond $0$ and $1$) are not covered, and these
may well be the most interesting cases (cf. matroids for the submodular case). 

\section{Basic $k$-submodular functions: Case analysis}

Finally, we go through the case analysis required to show that all the relaxations 
listed in the proof sketch of Lemma~\ref{lemma:usefulksubmod} are actually $k$-submodular.

\begin{proof}[Full proof of Lemma~\ref{lemma:usefulksubmod}.]
\emph{Case 1.} Let $f$ be a unary function of $\{1,\ldots,k\}$, and
$f'$ the relaxation to $\{0,\ldots,k\}$ as in the proof sketch.
Consider two domain values $x$ and $y$. If $x$ and $y$ are integral and distinct,
then $x \sqcap y=x \sqcup y=0$, and the inequality holds; otherwise,
the outputs $x \sqcap y$ and $x \sqcup y$ are a reordering of the inputs.

\emph{Case 2.} For the bijection case, let $f$ be the relaxation,
and consider two evaluations $f(x_1,y_1)$ and $f(x_2,y_2)$. 
We refer to $(x_1,y_1)$ and $(x_2,y_2)$ as the \emph{input}, 
and the tuples of the resulting right-hand-side (after application 
of $\sqcap$ and $\sqcup$) as the \emph{output}. We split the proof by the
number of variables $x_1, y_1, x_2, y_2$ that take the value zero. 
If none of them takes the value zero, then either the output equals the 
input, or the output is all-zero, or the output has one all-zero column
and the input costs at least $1$; all these satisfy the $k$-submodularity inequality.
If one input, say $(x_1,y_1)$, equals $(0,0)$, then the output
equals the input.

If exactly one value is zero, assume w.l.o.g.\ that $x_1=d$ and $x_2=0$;
the same two values occur in the output (in the first ``column''), 
and we note that the other two output values (the second ``column'') equal
each other. Thus either the output equals the input, or the output has an
all-zero column and cost \half, while the input costs at least as much. 

If $x_1=x_2=0$ but $y_1, y_2 \neq 0$ (or similarly with $x$ and $y$
swapped), then either the output equals the input, or the output has cost
zero. Finally, with two zero-values in different columns and tuples, the
input costs~$\half+\half$ and the output contains one tuple~$(0,0)$ at cost
zero. This finishes the case.

\emph{Case 3.}
Let $f_{a,b}$ be the function defined in the proof sketch; 
we show that it is~$k$-submodular. 

Refer to~$a$ in the first coordinate, or~$b$ in the second coordinate, as a safe
coordinate; note that~$f_{a,b}$ can be viewed as taking cost $0$ if at
least one coordinate is safe, and otherwise \half times the
number of non-safe integral coordinates. We split into cases. First,
assume that one column of the output contains two integral non-safe values. Then this
column must be constant in input and output. If the other output column contains two
zeros, then the output costs~$1$ and the input costs either at least~$1+0$ or~$\half+\half$.
With one zero, the output is a reordering of the input, and nothing is changed. With no
zeros, input and output are constant and identical. 

Second, assume that both output columns contain one non-safe integral value each. Then the
output is~$(0,0)$ and~$(d,d')$, where~$d$ and~$d'$ are non-safe, but then
the output columns are just reorderings of the input columns, so the input
costs either~$\half+\half$ or~$1+0$.  

In the last cases, the total number of non-safe integral values in the
output is either~$0$, at output cost zero, or~$1$. In the last case, the
maximum total output cost is~\half, in which case the non-safe 
column of the output is~$0, d$ for some~$d$, the parallel column is~$0,0$,
and the input contains 
either a tuple~$(d,0)$ or~$(0,d')$ for unsafe integral values~$d, d'$. 

\emph{Case 4.}
We show $k$-submodularity. Consider the total cost of the input. If the
input has total cost zero, then the output is either all-zero or identical to
the input. If the input has a tuple of cost zero, it must be constant,
say~$(x,\ldots,x)$. If~$x=0$, then the output equals the input; otherwise, the
output uses only values~$0$ and~$x$. The $\sqcap$-tuple contains~$x$ if and only
if~$x$ occurs in the other tuple; the~$\sqcup$-tuple contains~$0$ if and
only if some~$d' \notin \{0,d\}$ occurs in the other tuple. Each event
``costs'' at most~\half, and if both events occur, the input costs~$1$. 

If the input cost is~$\half+\half$, then there are similarly two essential
cases (the non-zero entries are identical or different), both of which
have an output of total cost at most~$1$.  Otherwise, the input costs at
least~$1+\half$, and the output can only cost~$1+1$ if there 
are two distinct constant non-zero columns in the input (in which case the
input costs~$1+1$). 
\maybeqed \end{proof}

\end{document}